\documentclass[letter,journal,final]{IEEEtran}  % for Arxiv twocolumn
\usepackage{arxiv}  
\usepackage{amsthm}  %  for proof environment
\newcommand{\publicationversion}{Preprint % or Author accepted version
} 

\usepackage{cite} 
\usepackage{amsmath,amssymb,amsfonts,nccmath}
\usepackage[lofdepth,lotdepth]{subfig}
\usepackage{algorithmic}
\usepackage{graphicx}
\usepackage{textcomp}
\usepackage{color,soul}
\usepackage{tikz}  
\usepackage[colorlinks=true,linkcolor=blue, citecolor=blue, urlcolor=black]{hyperref}

\newtheorem{thm}{Theorem}
\newtheorem{lem}{Lemma}
\newtheorem{prop}{Proposition}

\newtheorem{defn}{Definition}
\newtheorem{assum}{Assumption}

\newtheorem{rem}{Remark}

\DeclareMathOperator{\dom}{dom}  
\newcommand{\T}{^{\top}} 
\newcommand{\ie}{\textit{i.e.}}
\allowdisplaybreaks

\def\BibTeX{{\rm B\kern-.05em{\sc i\kern-.025em b}\kern-.08em
		T\kern-.1667em\lower.7ex\hbox{E}\kern-.125emX}}
\begin{document}
	\title{Hybrid Feedback for Affine Nonlinear Systems with Application to Global Obstacle Avoidance \\ (Extended Version)}
	\author{%Miaomiao Wang  and Abdelhamid Tayebi
		\href{https://orcid.org/0000-0001-7498-5162}{\includegraphics[scale=0.06]{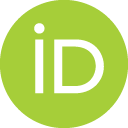}\hspace{0.5mm}} Miaomiao Wang \\
		Huazhong University of Science and Technology\\
		Wuhan, 430074, China \\
		\texttt{mmwang@hust.edu.cn}		
		\And			 	
		\href{https://orcid.org/0000-0002-9049-4651}{\includegraphics[scale=0.06]{orcid.png}\hspace{0.5mm}} Abdelhamid Tayebi \\
		Lakehead University\\
		Thunder Bay, ON P7B 5E1, Canada \\
		{\tt atayebi@lakeheadu.ca} 
		%\thanks{\copyrightNoticeSubmitted{IEEE}}
	}% 
	
	\twocolumn[ 
	\begin{@twocolumnfalse}
		\maketitle
		%\nocite{*}
		\begin{abstract}
			This paper explores the design of hybrid feedback for a class of affine nonlinear systems with topological constraints that prevent global asymptotic stability. A new hybrid control strategy is introduced, which differs conceptually from the commonly used synergistic hybrid approaches. The key idea involves the construction of a generalized synergistic Lyapunov function whose switching variable can either remain constant or dynamically change between jumps.
			Based on this new hybrid mechanism, a generalized synergistic hybrid feedback control scheme, endowed with global asymptotic stability guarantees, is proposed. This hybrid control scheme is then improved through a smoothing mechanism that eliminates discontinuities in the feedback term. Moreover, the smooth hybrid feedback is further extended to a larger class of systems through the integrator backstepping approach.
			The proposed hybrid feedback schemes are applied to solve the global obstacle avoidance problem using a new concept of synergistic navigation functions. Finally, numerical simulation results are presented to illustrate the performance of the proposed hybrid controllers.
		\end{abstract}
		
	\end{@twocolumnfalse}]
	\copyrightfootnote{This is the extended version of the work accepted for publication in IEEE TAC, DOI: \href{https://doi.org/10.1109/TAC.2024.3372463}{10.1109/TAC.2024.3372463}.}% 
	
	\section{Introduction}
	In the present work, we are interested in the design of feedback control, for a class of continuous-time affine nonlinear systems, with robust global asymptotic stability (GAS) guarantees. It is clear that this objective is not attainable for certain control systems evolving on a non-contractible state space \cite{sontag1998mathematical} or a state space that is not diffeomorphic to any Euclidean space \cite{bhatia1970stability}. This is referred to as the topological obstruction to GAS, and examples of such control problems include continuous feedback for the stabilization of the rotational motion on a compact boundaryless manifold \cite{bhat2000topological}, and the robot navigation in a sphere world \cite{koditschek1990robot}. This motivates the development of the hybrid control architectures whose feedback terms have both continuous-time and discrete-time behaviors \cite{goebel2012hybrid}.
	
	In order to overcome the aforementioned topological obstruction and achieve robust GAS, \emph{synergistic hybrid feedback}, designed from a synergistic family of potential functions, has proven to be a powerful strategy \cite{mayhew2011synergistic,mayhew2013synergistic}. Synergistic hybrid feedback has been successfully applied to various control problems for rigid body systems to achieve GAS results, including orientation stabilization on the $n$-dimensional sphere \cite{mayhew2013global,casau2019robust}, attitude control on Lie group $SO(3)$ \cite{berkane2016construction,wang2022hybrid,tong2023synergistic}, global tracking on compact manifolds \cite{casau2019hybrid}, and global trajectory tracking of underwater vehicles \cite{casau2015robust}.  
	
	%Hybrid feedback for navigation: 
	Another important application of hybrid feedback is the crucial and long-standing problem of safe robot navigation with global obstacle avoidance \cite{koditschek1990robot,paternain2017navigation,li2018navigation,arslan2019sensor,verginis2021adaptive}. As pointed out in \cite{koditschek1990robot}, there exists at least one saddle equilibrium point for each obstacle within the state space of a sphere world, which prevents global asymptotic stabilization of a set-point with time-invariant continuous feedback. Consequently, the best result one can achieve is almost global asymptotic stability (AGAS). To overcome these constraints, some hybrid control approaches for global robot navigation with, obstacle avoidance, have been proposed recently \cite{sanfelice2006robust,berkane2021obstacle,braun2021explicit,poveda2021robust,sawant2021hybrid,marley2021synergistic}.
	
	The first group of hybrid approaches relies on hybrid switching strategies between different control modes, such as the \textit{obstacle-avoidance} mode and \textit{move-to-target} mode in \cite{berkane2021obstacle,braun2021explicit,sawant2021hybrid}. The control laws for each mode are specially designed for particular purposes and are switched by a logical variable, ensuring GAS for the overall hybrid system. These hybrid switching strategies enable the handling of more general obstacles, such as arbitrary convex obstacles in \cite{sawant2021hybrid}. However, the application of these hybrid strategies was limited to linear systems \cite{braun2021explicit} and single-integrator systems \cite{berkane2021obstacle,sawant2021hybrid}. Another group of hybrid approaches extends the framework of synergistic hybrid feedback for a class of nonlinear dynamical systems  \cite{sanfelice2006robust,casau2019adaptive,marley2021synergistic,casau2022robust}. The hybrid feedback laws are derived from synergistic Lyapunov functions and are switched by a logical variable when the state is close to one of the undesired equilibria. Both hybrid feedback strategies rely on a switching mechanism through a logical variable. In the first (above-mentioned) approach the logical variable is used to switch between different control modes, while in the synergistic approach, it is used to switch between the synergistic Lyapunov functions.
	
	In recent years, there has been significant research devoted to the development of synergistic hybrid feedback relying on centrally/non-centrally synergistic potential functions \cite{mayhew2011hybrid,mayhew2013synergistic}, synergistic Lyapunov functions \cite{mayhew2011synergistic}, and synergistic control barrier functions \cite{marley2021synergistic}. Recently, a new hybrid control strategy, relying on \textit{dynamic} synergistic potential functions, has been introduced in the authors' previous work \cite{wang2019new,wang2022hybrid}. The key advantage of this approach is that it only requires a single dynamic synergistic potential function, which allows the switching variable to flow and not necessarily be constant between jumps. This hybrid approach has been successfully applied to solve the global tracking problem on Lie groups $SO(3)$ and $SE(3)$. More recently, some generalized synergistic hybrid feedback approaches for a broad class of dynamical systems, using the same idea of dynamically changing variables, have been proposed in \cite{schmidt2022generalization,casau2022robust}.
	
	In this paper, we propose a hybrid feedback control strategy, based on a new generalized synergistic Lyapunov function, for a broad range of affine nonlinear systems, with GAS guarantees. We also provide some tools to smooth out the discontinuities that arise from the jumps of the logical variable and extend our approach to a larger class of systems via integrator backstepping techniques. Finally, we apply the proposed hybrid feedback to the robot navigation problem with global obstacle avoidance.
	The main contributions of this paper are as follows:
	\begin{itemize}
		\item [1)]  Our proposed hybrid approach, stemming from our preliminary work \cite{wang2019new}, relies on a vector-valued switching variable that can evolve (dynamically) between jumps, and allows to expand the scope of applicability of synergistic hybrid feedback control strategies beyond the available hybrid approaches in the literature \cite{mayhew2011hybrid,mayhew2013synergistic,mayhew2011synergistic,casau2019adaptive}.
		\item [2)] We propose three types of synergistic hybrid feedback approaches, including a nominal non-smooth hybrid feedback, a hybrid feedback with smooth control input, and a backstepping-based hybrid feedback. These proposed approaches guarantee GAS for a broad class of affine nonlinear systems.
		
		\item [3)] We propose a smoothing mechanism to remove the discontinuities in the control input that arise from the jumps of the switching variable. Our proposed smoothing mechanism can handle a large class of hybrid feedback systems, and includes the approach in \cite{mayhew2011synergistic,mayhew2013synergistic} as a special case. 
		
		\item [4)] Our proposed hybrid feedback approaches have been successfully applied to the challenging problem of global obstacle avoidance for robot navigation with single integrators using synergistic navigation functions. This distinguishes our approach from the existing hybrid feedback approaches proposed in \cite{berkane2021obstacle,sawant2021hybrid,braun2021explicit,sanfelice2006robust,casau2019adaptive,casau2022robust}.
	\end{itemize}

	The remainder of this paper is organized as follows. Section \ref{sec:backgroud} introduces some preliminary notions that will be used throughout the paper. Section \ref{sec:Hybrid_feedback_GAS} presents the design of the proposed synergistic hybrid feedback. In Section \ref{sec:navigation}, we provide an application of the hybrid feedback for robot navigation.

	%%%%%%%%%%%%%%%%%%%%%%%%%%%%%%%%%%%%%%%%%%%%%%%%%%%%%%%%%%%%%%%%%%%%%%%%%%%%%%%%
	\section{Preliminaries} \label{sec:backgroud}
	\subsection{Notations and Definitions} 
	The sets of real, non-negative real, positive real and natural numbers are denoted by $\mathbb{R}$, $\mathbb{R}_{\geq 0}$, $\mathbb{R}_{> 0}$ and $\mathbb{N}$, respectively. We denote by $\mathbb{R}^n$ the $n$-dimensional Euclidean space. The inner product in $\mathbb{R}^n$ is denoted by $\langle \cdot,\cdot \rangle$ and the Euclidean norm of a vector $x\in \mathbb{R}^n$ is defined as $\|x\|:=\sqrt{\langle x,x \rangle}$. The $n$-by-$n$ identity matrix is denoted by $I_n$.  
	The closure of a subset $\mathcal{X}\subset \mathbb{R}^n$ is denoted by $\overline{\mathcal{X}}$.
	We define $c \oplus r\mathbb{B}: = \{x\in \mathbb{R}: \|x-c\|\leq r\}$ as a closed ball centered at $c$ with radius $r>0$, and $\mathcal{O} \oplus r \mathbb{B}:= \bigcup_{c\in \mathcal{O}} c \oplus r \mathbb{B}$ with $\mathcal{O}\subset \mathbb{R}^n$.  
	For a closed subset $\mathcal{X}\subset \mathbb{R}^n$, the \textit{gradient} of a differentiable smooth function $f:\mathcal{X} \to  \mathbb{R}$ at point $x\in \mathcal{X}$ is denoted by $\nabla_x f(x)$. 
	A point $x\in \mathcal{X}$ is called a \textit{critical point}  of $f$ if the gradient varnishes at point $x$ (\ie, $\nabla_{x} f(x) = 0$).
	The \textit{Jacobian matrix} of a smooth function $F:\mathbb{R}^n \to \mathbb{R}^m$ relative to $x$ is defined by %$\mathcal{D}_x F(x) :=  \frac{\partial F(x)}{\partial x\T} = \big[\frac{\partial F_i(x)}{\partial x_j}\big]_{1\leq i\leq m, 1\le j\leq n}\in \mathbb{R}^{m\times n}$
	$$
	\mathcal{D}_x F(x) :=  \frac{\partial F(x)}{\partial x\T} = \begin{bmatrix}
		\frac{\partial F_1}{\partial x_1}(x) & \cdots & \frac{\partial F_1}{\partial x_n}(x) \\
		\vdots                      &        & \vdots \\
		\frac{\partial F_m}{\partial x_1}(x) & \cdots & \frac{\partial F_m}{\partial x_n}(x)
	\end{bmatrix}  \in \mathbb{R}^{m\times n}
	$$
	with $F(x) = [F_1(x),\dots,F_m(x)]\T \in \mathbb{R}^m$ and $x=[x_1,\dots,x_n]\T \in \mathbb{R}^n$.
	If $f:\mathbb{R}^n \to \mathbb{R}$, then one has $\nabla_x f(x) = \mathcal{D}_x f(x)\T$ for all $x\in \mathbb{R}^n$. If $f$ has multiple arguments as $f:\mathbb{R}^n \times \mathbb{R}^{r} \to \mathbb{R}$, then one has $\nabla_x f(x,y) = \mathcal{D}_x f(x,y)\T$ and $\nabla_y f(x,y) = \mathcal{D}_y f(x,y)\T$ for all $(x,y)\in \mathbb{R}^n \times \mathbb{R}^{r}$.  Unless we explicitly specify otherwise, the time argument of the time-dependent variables is always omitted.
	\begin{defn}[Bouligand's tangent cone\cite{bouligand1932introduction}]
		The \textit{tangent cone} to a closed set $\mathcal{X}\subset \mathbb{R}^n$ at a point $x\in \mathcal{X}$ is the subset of $\mathbb{R}^n$ defined by
		\begin{align}%\mathsf{T}_x\mathcal{X}
			\mathsf{T}_\mathcal{X}(x) = \left\{z: \liminf_{\tau \to 0} \frac{d_{\mathcal{X}}(x + \tau z)}{\tau} = 0\right\}
		\end{align}
		where $\mathsf{d}_\mathcal{X}(x): = \inf_{y\in \mathcal{X}}\|y-x\|$ denotes the distance from a point $x\in \mathbb{R}^n$ to the set $\mathcal{X}$.  
	\end{defn} 
	The tangent cone $\mathsf{T}_\mathcal{X}(x)$ is non-trivial only on the boundary of $\mathcal{X}$ since $\mathsf{T}_\mathcal{X}(x) = \mathbb{R}^n$ if $x$ is in the interior of $\mathcal{X}$, and $\mathsf{T}_\mathcal{X}(x) = \emptyset$ if $x \notin \mathcal{X}$. Consider the system $\dot{x}(t) = f(x(t))$ and assume that, for each initial condition $x(0)$ in an open set $\mathcal{O}$, it admits a unique solution defined for all $t\geq 0$. Then, according to Nagumo’s theorem \cite{nagumo1942lage}, the closed set $\mathcal{X}\subset \mathcal{O}$ is forward (positively) invariant if and only if $f(x) \in \mathsf{T}_\mathcal{X}(x)$ for all $x\in \mathcal{X}$.

	\subsection{Hybrid Systems Framework} 
	We consider the following hybrid system with state space $\mathbb{R}^n$ defined in \cite{goebel2012hybrid} as:
	\begin{equation}\mathcal{H}:
		\begin{cases}
			\dot{x} ~~\in F(x), & \quad x \in \mathcal{F} \\
			x^{+} \in G(x),     & \quad x \in \mathcal{J}
		\end{cases} \label{eqn:hybrid_system}
	\end{equation}
	where the \textit{flow map} $F: \mathbb{R}^n \to \mathbb{R}^n$ describes the continuous flow of $x$ on the \textit{flow set} $\mathcal{F} \subseteq \mathbb{R}^n$; the \textit{jump map} $G: \mathbb{R}^n\rightrightarrows  \mathbb{R}^n$ (a set-valued mapping from $\mathbb{R}^n$ to $\mathbb{R}^n$) describes the discrete flow of $x$ on the \textit{jump set} $\mathcal{J} \subseteq \mathbb{R}^n$. Define the \textit{hybrid time domain} as a subset $E \subset \mathbb{R}_{ \geq 0} \times \mathbb{N}$ in the form of
	$ E = \bigcup_{j=0}^{J-1} ([t_j,t_{j+1}] \times \{j\}),	$
	for some finite sequence $0=t_0 \leq t_1 \leq \cdots \leq t_J$, with the last interval possibly in the form $[t_{J}, T)$ with $T$ finite or $T=+\infty$.  
	A hybrid arc is a function $x: \dom{x} \to \mathbb{R}^n$, where $\dom{x}$ is a hybrid time domain and, for each fixed $j\in \mathbb{N}_{>0}$, $t \mapsto x(t,j)$ is a locally absolutely continuous function on the interval $\mathcal{I}_j = \{t:(t,j) \in \dom{x}\}$. Note that $x^+$  denotes the value of $x$ after a jump at the current time, namely, $x^+(t,j)=x(t,j+1)$.  
	A solution $x$ to $\mathcal{H}$ is said to be  \textit{maximal} if it cannot be extended by flowing nor jumping, and \textit{complete} if its domain $\dom x$ is unbounded. Let $|x|_{\mathcal{A}}$  denote the distance of a point $x$ to a closed set $\mathcal{A} \subset \mathbb{R}^n$. The set $\mathcal{A}$ is said to be: \textit{stable} for $\mathcal{H}$ if for each $\epsilon>0$ there exists $\delta>0$ such that each maximal solution $x$ to $\mathcal{H}$ with $|x(0,0)|_{\mathcal{A}} \leq \delta$ satisfies $|x(t,j)|_{\mathcal{A}} \leq \epsilon$ for all $(t,j)\in \dom x$; \textit{globally attractive} for $\mathcal{H}$ if  every maximal solution $x$ to $\mathcal{H}$  is complete and satisfies $\lim_{t+j\to \infty}|x(t,j)|_{\mathcal{A}} = 0$ for all $(t,j)\in \dom x$; \textit{globally asymptotically stable} if it is both stable and globally attractive for $\mathcal{H}$. For more details on dynamic hybrid systems, we refer the readers to \cite{goebel2012hybrid} and references therein.
	
	\subsection{Problem Statement and Motivation}\label{sec:problem}
	Let $\mathcal{X}$ be a closed nonempty subset of $\mathbb{R}^n$. Consider the affine nonlinear system of the form
	\begin{align}
		\dot{x} & = f(x) + g(x)u \label{eqn:affine_system}
	\end{align}
	where $x\in \mathcal{X}$ is the state, $u\in \mathbb{R}^m$ is the control input, the \textit{drift term} $f:\mathcal{X} \to \mathbb{R}^n$ and \textit{input gain} $g: \mathcal{X} \to \mathbb{R}^{n\times m}$ are smooth. The goal is to design a control scheme for system \eqref{eqn:affine_system} such that the (compact) equilibrium set $\mathcal{A}_o \subset \mathcal{X}$ is GAS.

	Consider the real-valued functions  $V_o:\mathcal{X} \to \mathbb{R}_{\geq 0}$ and $\kappa_o:\mathcal{X} \to \mathbb{R}^m$. Let $V_o$ be a Lyapunov function candidate with respect to the set $\mathcal{A}_o$, and $u = \kappa_o(x)$ be a smooth time-invariant feedback for system \eqref{eqn:affine_system} satisfying:
	\begin{align*}
		\left.
		\begin{array}{rl}
			f(x) + g(x)\kappa_o(x)                                 & \in \mathsf{T}_\mathcal{X}(x) \\
			\langle \nabla_x V_o(x), f(x) + g(x)\kappa_o(x)\rangle & \leq 0
		\end{array} \right\} ~~ \forall x\in \mathcal{X} .  %\label{eqn:prop_kappa_o}
	\end{align*} 
	Let $\mathcal{E}_o:=\{x\in \mathcal{X}: \langle  \nabla_x V_o(x), f(x) + g(x)\kappa_o(x)\rangle =  0\} $
	and $\Psi_{V_o}\subseteq \mathcal{E}_o$ denote the largest invariant set  for the closed-loop system %$\dot{x} = f(x) + g(x)\kappa_o(x)$.
	$ \dot{x}   = f(x) + g(x)\kappa_o(x)$ with $x\in  \mathcal{E}_o$.
	Then, from \cite[Corollary 4.2]{khalil2002nonlinear}, one concludes that all the solutions to the closed-loop system must converge to the set $\Psi_{V_o}$. Consequently, the set $\mathcal{A}_o$ is GAS only if $\mathcal{A}_o=\Psi_{V_o}$. However, it is challenging to achieve GAS using the feedback $\kappa_o(x)$ when $\mathcal{A}_o$ is only a subset of the largest invariant set $\Psi_{V_o}$ (\ie, $\mathcal{A}_o \subset \Psi_{V_o}$).

	We introduce a finite non-empty set $\Theta \subset  \mathbb{R}^r$ and a vector-valued switching variable $\theta\in  \mathbb{R}^r$ with the following hybrid dynamics inspired by \cite{wang2022hybrid}:
	\begin{align}
		\mathcal{H}_\theta : \begin{cases}
			\dot{\theta} ~~ = \varpi(x,\theta), & (x,\theta)\in \mathcal{F} \\
			\theta^+  \in G_o(x,\theta),          & (x,\theta)\in \mathcal{J}
		\end{cases} \label{eqn:theta_hybrid_dynamics}
	\end{align}
	where $\mathcal{F}, \mathcal{J} \subset \mathcal{X} \times \mathbb{R}^r$ denote the flow and jump sets, respectively, and $\varpi:\mathcal{X} \times \mathbb{R}^r \to \mathbb{R}^r$ and $G_o: \mathcal{X} \times \mathbb{R}^r \rightrightarrows  \Theta \subset  \mathbb{R}^r$ denote the flow and jump maps, respectively. Note that the switching variable $\theta$ can be either constant or dynamically changing over the flows, depending on its flow dynamics $\varpi$. In \cite{mayhew2011synergistic}, the variable $\theta \in \mathbb{Z}$ is considered as a logical index variable with $ \varpi(x,\theta)\equiv 0$.
	In \cite{wang2022hybrid}, the variable $\theta$ is considered as a continuous-discrete scalar variable with $\varpi(x,\theta)$ generated from the gradient of a dynamic synergistic potential function. 
	In the present work, $\theta$ is a vector-valued switching variable that flows via the flow map $\varpi(x,\theta)$ and jumps via the jump map $G_o(x,\theta)$.  
	With an appropriate design of the feedback term and the switching mechanism, the state is kept away from the undesired equilibrium points and driven towards the desired equilibrium.

	\begin{figure}[!ht]
		\centering
		\includegraphics[width=0.85\linewidth]{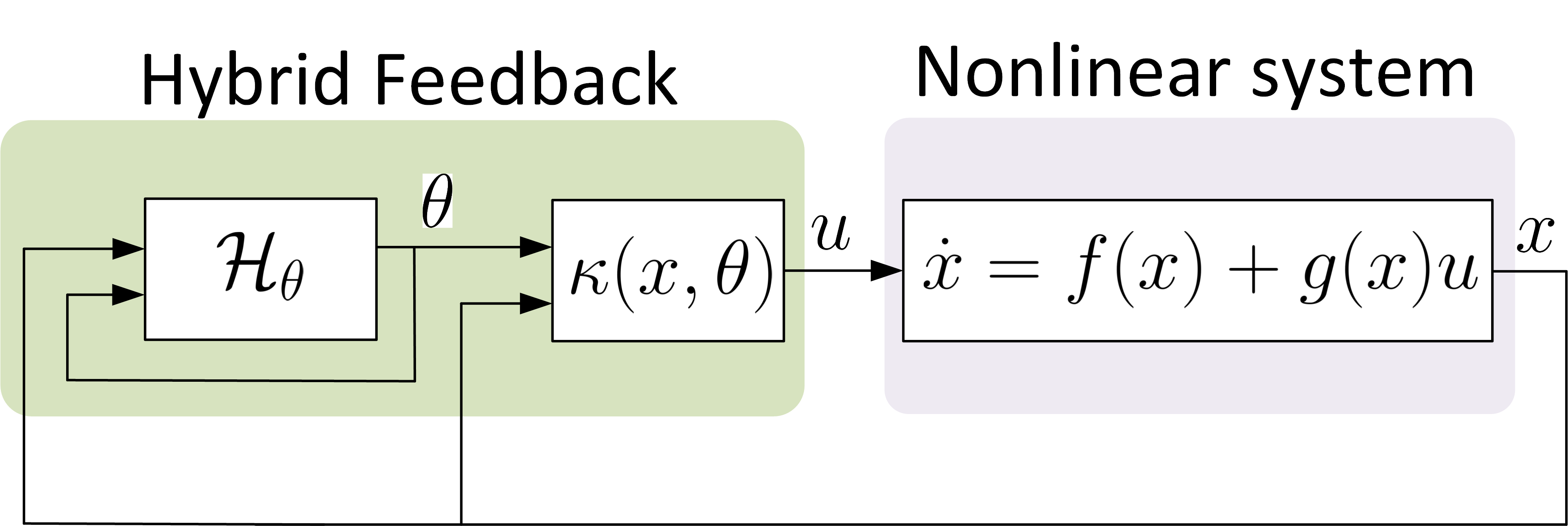}
		\caption{The architecture of the proposed hybrid feedback strategy.}
		\label{fig:diagram31}
	\end{figure}
	Without loss of generality, we consider the following extended affine nonlinear system modified from \eqref{eqn:affine_system} and \eqref{eqn:theta_hybrid_dynamics}:
	\begin{align}\label{eqn:new_affine_system}
		\begin{pmatrix}
			\dot{x} \\
			\dot{\theta}
		\end{pmatrix} =  \underbrace{\begin{pmatrix}
				f(x) \\
				\varpi(x,\theta)
		\end{pmatrix}}_{f_c(x,\theta)} + \underbrace{\begin{pmatrix}
				g(x) \\
				0
		\end{pmatrix}}_{g_c(x,\theta)} u
	\end{align}
	with $(x,\theta)\in \mathcal{X}\times \mathbb{R}^r$ denoting the new extended state, $u \in \mathbb{R}^m$ denoting the control input to be designed.
	The structure of our hybrid control strategy is given in Fig \ref{fig:diagram31}. 
	Then, given the switching variable $\theta$ and its hybrid dynamics $\mathcal{H}_\theta:=\{\varpi, G_o,\mathcal{F},\mathcal{J}\}$ in \eqref{eqn:theta_hybrid_dynamics}, the objective of this work is to design a hybrid feedback $u=\kappa(x,\theta)$ for the modified system \eqref{eqn:new_affine_system} such that the compact set $\mathcal{A} =\mathcal{A}_o \times \mathcal{A}_\theta$ (with $\mathcal{A}_\theta\subset \mathbb{R}^r$) is GAS for the overall hybrid closed-loop system.

	\section{Hybrid Feedback with GAS Guarantees}\label{sec:Hybrid_feedback_GAS}
	\subsection{Synergistic Hybrid Feedback}
	
	Consider the smooth functions $V: \mathcal{X} \times \mathbb{R}^r \to \mathbb{R}_{\geq 0}, \varpi:\mathcal{X} \times \mathbb{R}^r \to \mathbb{R}^r, \kappa:\mathcal{X}\times \mathbb{R}^r \to \mathbb{R}^m$, and a finite non-empty set $\Theta \subset \mathbb{R}^r$. Define the function $\mu_{V,\Theta}:\mathcal{X} \times \mathbb{R}^r \to \mathbb{R}$ as
	\begin{align}
		\mu_{V,\Theta}(x,\theta) & :=  V(x,\theta) - \min_{\bar{\theta}\in \Theta} V(x,\bar{\theta}) \label{eqn:mu_V}
	\end{align}
	and the set
	\begin{multline}
		%		\mathcal{E}:= \{(x,\theta)\in \mathcal{X} \times \mathbb{R}: & L_f V(x) + L_g V(x)\kappa(x,\theta) 
		%		\nonumber \\
		%		& = 0, \nabla_\theta V(x,\theta) = 0\}
		\mathcal{E}:=  \{(x,\theta)\in \mathcal{X} \times \mathbb{R}^r:\\
		\langle \nabla V(x,\theta), f_c(x,\theta) + g_c(x,\theta)\kappa(x,\theta)\rangle = 0 \} \label{eqn:def_E}
	\end{multline}
	with  $\nabla V(x,\theta) = [\nabla_x V(x,\theta)\T, \nabla_\theta V(x,\theta)\T]\T$ denoting the gradient of $V$. Let $\Psi_V \subseteq \mathcal{E} $ denote the  largest weakly invariant set
	\footnote{For a hybrid system $\mathcal{H}$, the set $S \subset \mathbb{R}^n$ is said to be weakly invariant if it is both weakly forward invariant and weakly backward invariant \cite[Definition 6.19]{goebel2012hybrid}. The descriptor ``weakly'' indicates that at least one (complete) solution to $\mathcal{H}$ is required to remain in the set $S$.}
	for system \eqref{eqn:new_affine_system} with $(x,\theta) \in  \mathcal{E}$. 
	Inspired by \cite{mayhew2011synergistic}, we introduce the following definition of synergistic feedback quadruple.
	\begin{defn}\label{defn:synergistic_feedback}
		Consider the real-valued functions $V:\mathcal{X} \times \mathbb{R}^r \to \mathbb{R}_{\geq 0}$,  $\kappa:\mathcal{X} \times \mathbb{R}^r \to \mathbb{R}^m$, $\varpi:\mathcal{X} \times \mathbb{R}^r \to \mathbb{R}^r$, and the non-empty set $\Theta \subset \mathbb{R}^r$. The quadruple $(V,\kappa,\varpi, \Theta)$ is said to be a \textit{synergistic feedback quadruple} relative to a compact set $\mathcal{A} \subset \mathcal{X} \times \mathbb{R}^r $, for system \eqref{eqn:new_affine_system}, with gap exceeding $\delta$, if
		\begin{itemize}
			\item [C1)] for all $\epsilon\geq 0$, the sub-level set  $\mho_V(\epsilon):= \{(x,\theta)\in \mathcal{X} \times \mathbb{R}^r: V(x,\theta)\leq \epsilon\} $ is compact;
			\item [C2)] $V$ is positive definite with respect to $\mathcal{A}$;
			\item [C3)] for all $(x,\theta)\in \mathcal{X} \times \mathbb{R}^r$, one has
			\begin{align}
				\langle \nabla V(x,\theta), f_c(x,\theta) + g_c(x,\theta)\kappa(x,\theta)\rangle  \leq 0.  \label{eqn:prop_kappa}
			\end{align}
			\item [C4)] there exists a constant $\delta>0$ such that %for all $(x,\theta)\in \Psi_V \setminus \mathcal{A}$ one has
			\begin{align}
				\delta_{V,\Theta} & := \inf_{(x,\theta)\in \Psi_V \setminus \mathcal{A}}  \mu_{V,\Theta}(x,\theta)   > \delta. \label{eqn:mu_V_delta}
			\end{align}
		\end{itemize}
	\end{defn}

	Let $\mathcal{A} :=\mathcal{A}_o \times \mathcal{A}_\theta$ be a compact set with $\mathcal{A}_o \subset \mathcal{X}, \mathcal{A}_\theta\subset \mathbb{R}^r$, and $(V,\kappa,\varpi,\Theta)$ be a synergistic feedback quadruple relative to $\mathcal{A}$ for system \eqref{eqn:new_affine_system} with  gap exceeding $\delta$.  
	We propose the following hybrid control feedback for system \eqref{eqn:affine_system}:
	\begin{align}
		\underbrace{
			\begin{array}{ll}
				u            & =  \kappa(x,\theta) \\
				\dot{\theta} & = \varpi(x,\theta)
			\end{array}~
		}_{(x,\theta)\in \mathcal{F}}
		\underbrace{
			\begin{array}{ll}
				&\\
				\theta^+ & = G_o(x,\theta)
			\end{array}~
		}_{(x,\theta)\in \mathcal{J}} \label{eqn:hybrid_feedback1}
	\end{align}
	where the set-valued map $G_o$ and the flow and jump sets $\mathcal{F},\mathcal{J}$ are defined as
	\begin{subequations}\label{eqn:hybrid_feedback1_design}
		\begin{align}
			%u           & = \kappa(x,\theta)\\
			%\varpi(x,\theta) & =- k_\theta\nabla_\theta V(x,\theta)  \label{eqn:def_f_theta}\\
			G_o(x,\theta) & := %\{ \bar{\theta} \in \Theta: \mu_{V,\Theta}(x,\bar{\theta})= 0 \} 
			\arg \min\{V(x,\bar{\theta}): \bar{\theta} \in \Theta\}\label{eqn:def_G_theta} \\
			\mathcal{F} & := \{(x,\theta)\in \mathcal{X}\times \mathbb{R}^r: \mu_{V,\Theta}(x,\theta) \leq \delta\} \label{eqn:def_F_set}\\
			\mathcal{J} & := \{(x,\theta)\in \mathcal{X}\times \mathbb{R}^r: \mu_{V,\Theta}(x,\theta) \geq \delta\}  \label{eqn:def_J_set}
		\end{align}
	\end{subequations}
	with $\mu_{V,\Theta}$ defined in \eqref{eqn:mu_V}. From \eqref{eqn:new_affine_system}, \eqref{eqn:hybrid_feedback1} and \eqref{eqn:hybrid_feedback1_design}, one obtains the following hybrid closed-loop system:
	\begin{subequations}\label{eqn:hybrid_closed_loop1}
		\begin{align}
			\begin{pmatrix}
				\dot{x} \\
				\dot{\theta}
			\end{pmatrix}                  &  = F(x,\theta) : 
			= \begin{pmatrix}
				f(x) +  g(x) \kappa(x,\theta) \\
				\varpi(x,\theta)
			\end{pmatrix}, 
			& (x,\theta)\in \mathcal{F} \label{eqn:hybrid_closed_F} \\
			\begin{pmatrix}
				x^+ \\
				\theta^+
			\end{pmatrix}                  & \in    G(x,\theta) := 
			\begin{pmatrix}
				x \\
				G_o(x,\theta)
			\end{pmatrix},                  & (x,\theta)\in \mathcal{J} \label{eqn:hybrid_closed_G}
		\end{align}
	\end{subequations}
	with $(G_o, \mathcal{F}, \mathcal{J})$ defined in \eqref{eqn:hybrid_feedback1_design}.
	The next lemma shows that the hybrid closed-loop system \eqref{eqn:hybrid_closed_loop1} satisfies the hybrid basic conditions  \cite[Assumption 6.5]{goebel2012hybrid} guaranteeing the well-posedness of the hybrid closed-loop system:
	\begin{lem}\label{lem:HBC}
		The hybrid closed-loop system \eqref{eqn:hybrid_closed_loop1} satisfies the following conditions:
		\begin{itemize}
			\item [A1)] $\mathcal{F}$ and $\mathcal{J}$ are closed subsets of $\mathbb{R}^{n}$; 
			\item [A2)] $F$ is outer semicontinuous and locally bounded relative to $\mathcal{F}$, $\mathcal{F} \in \dom F$ , and $F(x)$ is convex for every $x \in \mathcal{F}$;
			%\item $F$ is continuous on $\mathcal{F}$
			\item [A3)] $G$ is outer semicontinuous and locally bounded relative to $\mathcal{J}$ and $\mathcal{J} \in \dom G$.
		\end{itemize}
	\end{lem}
	\begin{proof}
		See Appendix \ref{sec:HBC}.     
	\end{proof}
	
	Now, one can state one of our main results:
	\begin{thm}\label{thm:thm1}
		Suppose that $(V,\kappa,\varpi,\Theta)$ is a synergistic feedback quadruple relative to the compact set $\mathcal{A}$ for system \eqref{eqn:new_affine_system} with synergy gap exceeding $\delta>0$. Then, if $f(x) + g(x)\kappa(x,\theta) \in \mathsf{T}_\mathcal{X}(x)$ holds for all $(x,\theta)\in \mathcal{X}\times \mathbb{R}^r$,  
		the set $\mathcal{A}$ is GAS for the hybrid closed-loop system \eqref{eqn:hybrid_closed_loop1}.
	\end{thm}
	\begin{proof}
		Given the synergistic feedback quadruple $(V,\kappa,\varpi, \Theta)$ relative to the   set $\mathcal{A}$ with  gap exceeding $\delta$, it follows from Definition \ref{defn:synergistic_feedback} that for all $(x,\theta)\in \mathcal{F}$
		\begin{align}
			\dot{V}(x,\theta)
			& = \langle \nabla V(x,\theta), f_c(x,\theta) + g_c(x,\theta)\kappa(x,\theta)\rangle \nonumber \\
			& = \langle \nabla_x V(x,\theta), f(x) + g(x)\kappa(x,\theta)\rangle \nonumber\\
			& \quad  + \langle  \nabla_\theta V(x,\theta), \varpi(x,\theta) \rangle                        %\nonumber \\ 
			%	&= \langle \nabla_x V(x,\theta), f(x) + g(x)\kappa(x,\theta)\rangle \nonumber \\
			%	& ~~~  - k_\theta | \nabla_\theta V(x,\theta)|^2 
			\leq 0 . \label{eqn:dot_V}
		\end{align}
		Thus, $V$ is non-increasing along the flows of the hybrid system \eqref{eqn:hybrid_closed_loop1}.
		Let  $(x^+,\theta^+)$ denote the state after each jump. It follows from \eqref{eqn:hybrid_closed_G} that $x^+=x$ and $\theta^+\in G_o(x,\theta)$. Then, from the definitions of the jump map $G_o$ in \eqref{eqn:def_G_theta} and the jump set $\mathcal{J}$ in \eqref{eqn:def_J_set}, for each $(x,\theta)\in \mathcal{J}$, one obtains %$V(x^+,\theta^+) = - \min_{\bar{\theta}\in \Theta} V(x,\bar{\theta})$ and 
		\begin{align}
			V(x^+,\theta^+) & = V(x,G_o(x,\theta)) \nonumber\\
			& =  \min_{\bar{\theta}\in \Theta} V(x,\bar{\theta}) \nonumber \\
			& = V(x,\theta) - \mu_{V,\Theta}(x,\theta) \nonumber\\
			& \leq  V(x,\theta) - \delta \label{eqn:V^+}
		\end{align}
		where we made use of the facts $V(x,G_o(x,\theta)) = \min_{\bar{\theta}\in \Theta} V(x,\bar{\theta})$ and $\min_{\bar{\theta}\in \Theta} V(x,\bar{\theta}) =V(x,\theta) - \mu_{V,\Theta}(x,\theta)$ from the functions of $G_o$ in \eqref{eqn:def_G_theta} and  $\mu_{V,\Theta}$ in \eqref{eqn:mu_V}. It follows from \eqref{eqn:V^+} that $V$ is strictly decreasing after each jump of system \eqref{eqn:hybrid_closed_loop1}. Therefore, one concludes that $V$ is non-increasing for all $(x,\theta)\in \mathcal{F}$ and is strictly decreasing for all $(x,\theta)\in \mathcal{J}$.
		
		Next, we are going to show the global attractivity of the set $\mathcal{A}$. 
		Inspired by \cite{goebel2012hybrid}, we introduce the following functions:
		\begin{align}
			u_c(x,\theta) &= \begin{cases}
				\langle \nabla V(x,\theta), f_c(x,\theta) + g_c(x,\theta)\kappa(x,\theta)\rangle  \\
				\hfill  \text{ if } (x,\theta)\in \mathcal{F} \\
				-\infty \hfill \text{ otherwise} 
			\end{cases} \label{eqn:def_u_c} \\
			u_d(x,\theta) &= \begin{cases}
				-\delta & \text{ if } (x,\theta)\in \mathcal{J} \\
				-\infty & \text{ otherwise } 
			\end{cases} \label{eqn:def_u_d}
		\end{align}
		From \eqref{eqn:dot_V} and \eqref{eqn:V^+}, the growth of $V$ is upper bounded during flows by $u_c(x,\theta)\leq 0$ and during jumps by $u_d(x,\theta)\leq 0$ for each $(x,\theta)\in \mathcal{X}\times \mathbb{R}^r$. 
		It follows from \cite[Theorem 8.2]{goebel2012hybrid} that any maximal solution to \eqref{eqn:hybrid_closed_loop1} must converge to the largest weakly invariant subset of
		\begin{align}
			V^{-1}(r) \cap \big(\overline{u_c^{-1}(0)} \cup (u_d^{-1}(0)\cap G(u_d^{-1}(0))) \big) \label{eqn:invariant_set}
		\end{align} 
		for some $r\in \mathbb{R}$ in the image of $V$ with $V^{-1}(r):=\{(x,\theta)\in \mathcal{X}\times \mathbb{R}^r: V(x,\theta) = r\}$, $u_c^{-1}(0): = \{(x,\theta)\in \mathcal{F}: u_c(x,\theta) = 0\}$ and $u_d^{-1}(0): = \{(x,\theta)\in \mathcal{J}: u_d(x,\theta) = 0\}$. From \eqref{eqn:def_u_d}, it is obvious to show that $u_d^{-1}(0) = \emptyset$. Then, \eqref{eqn:invariant_set} can be reduced to %$V^{-1}(r) \cap   \overline{u_c^{-1}(0)}$, 
		\begin{align*}
			V^{-1}(r) \cap   \overline{u_c^{-1}(0)} % \label{eqn:invariant_set2}
		\end{align*} 
		which corresponds to $ \mathcal{F} \cap \mathcal{E}$ in view of the fact that $u_c^{-1}(0)=\mathcal{E}$ from \eqref{eqn:def_E} and \eqref{eqn:def_u_c}. Hence, according to the definition of the set $\Psi_V\subseteq \mathcal{E}$, any maximal solution to \eqref{eqn:hybrid_closed_loop1} must converge to the largest weakly invariant subset contained in $\mathcal{F} \cap\Psi_V$. 
		
		From the definitions of  $\mathcal{F}$ and $\mathcal{J}$ in \eqref{eqn:hybrid_feedback1_design}, one has $\mathcal{F} \cap \mathcal{A} = \mathcal{A}$ % and $\mathcal{J} \cap \mathcal{A} = \emptyset$ 
		since $\mu_{V,\Theta}(x,\theta)  = V(x,\theta) - \min_{\bar{\theta}\in \Theta} V(x,\bar{\theta}) = - \min_{\bar{\theta}\in \Theta} V(x,\bar{\theta}) \leq 0 < \delta $ for all $(x,\theta)\in \mathcal{A}$. It follows from $\mathcal{A} \subset \Psi_V$ and $\mathcal{F} \cap \mathcal{A} = \mathcal{A}$ that $\mathcal{A} \subset \mathcal{F}\cap \Psi_V$. Moreover, from \eqref{eqn:mu_V_delta}, one has $\mathcal{F}\cap (\Psi_V \setminus \mathcal{A}) = \emptyset$ %$(\Psi_{V} \setminus \mathcal{A}) \subset \mathcal{J}$.
		since $\mu_{V,\Theta}(x,\theta)> \delta$ for all $\Psi_V \setminus \mathcal{A}$.   
		Applying some simple set-theoretic arguments, one obtains
		\begin{align*}
			\mathcal{F} \cap \Psi_V & \subset \mathcal{F} \cap ((\Psi_V \setminus \mathcal{A}) \cup   \mathcal{A}) \nonumber\\
			& = (\mathcal{F} \cap (\Psi_V \setminus \mathcal{A})) \cup (\mathcal{F} \cap \mathcal{A}) \nonumber \\
			& = \emptyset \cup  \mathcal{A} = \mathcal{A}
		\end{align*}
		where we made use of the facts $\mathcal{F} \cap \mathcal{A} = \mathcal{A}$ and $\mathcal{F}\cap (\Psi_V \setminus \mathcal{A}) = \emptyset$.
		Hence, from $\mathcal{A} \subset \mathcal{F}\cap \Psi_V$ and $\mathcal{F} \cap \Psi_V\subset\mathcal{A}$, it follows that $ \mathcal{F}\cap \Psi_V = \mathcal{A}$. Therefore, one can conclude that any maximal solution to \eqref{eqn:hybrid_closed_loop1} must converge to $\mathcal{A}$. 
		As per Lemma \ref{lem:HBC}, the hybrid basic conditions \cite[Assumption 6.5]{goebel2012hybrid} are satisfied for the hybrid closed-loop system \eqref{eqn:hybrid_closed_loop1}. 
		Since $f(x) + g(x)\kappa(x,\theta) \in \mathsf{T}_{\mathcal{X}}(x) $ for all $(x,\theta)\in \mathcal{X}\times \mathbb{R}^r$ and $\varpi(x,\theta)\in \mathbb{R}^r$ by assumption, the closed set $\mathcal{X}\times \mathbb{R}^r$ is forward (positively) invariant. From \eqref{eqn:dot_V}, \eqref{eqn:V^+} and Definition \ref{defn:synergistic_feedback}, every maximal solution to \eqref{eqn:hybrid_closed_loop1} is bounded. Moreover, according to \eqref{eqn:V^+} there is no infinite sequence of jumps in a finite time interval (absence of Zeno behavior) for system \eqref{eqn:hybrid_closed_loop1}, and the number of jumps is bounded by the initial conditions (\ie, $J \leq V(x(0,0),\theta(0,0))/\delta$). From the definitions of $\mathcal{F}$ and $\mathcal{J}$ in \eqref{eqn:hybrid_feedback1_design}, one verifies $G(\mathcal{J}) \subset \mathcal{X}\times \mathbb{R}^r = \mathcal{F}\cup \mathcal{J}$. Therefore, by virtue of \cite[Proposition 6.10]{goebel2012hybrid},  it follows that every maximal solution to \eqref{eqn:hybrid_closed_loop1} is complete. Consequently, one concludes that the set $\mathcal{A}$ is GAS for the hybrid closed-loop system \eqref{eqn:hybrid_closed_loop1}, which completes the proof.
	\end{proof} 
	
	\begin{rem}
		The key idea behind the design of the flow and jump maps $(\mathcal{F},\mathcal{J})$ in \eqref{eqn:hybrid_feedback1_design} is to ensure that the desired equilibrium set $\mathcal{A}$  belongs (exclusively) to the flow set $\mathcal{F}$ (\ie, $ \mathcal{A} \subset \mathcal{F},~ \mathcal{J} \cap \mathcal{A} = \emptyset$), and the undesired equilibrium set, denoted by $\Psi_{V} \setminus \mathcal{A}$, belongs (exclusively) to the jump set $\mathcal{J}$ (\ie, $
		(\Psi_{V} \setminus \mathcal{A}) \subset \mathcal{J},~ \mathcal{F} \cap (\Psi_{V} \setminus \mathcal{A}) = \emptyset $). Moreover, the design of the jump map $G_o$ in \eqref{eqn:hybrid_feedback1_design} guarantees that $V$ has a strict decrease after each jump (\ie, $V(x^+,\theta^+\in G_o(x,\theta)) = \min_{\bar{\theta}\in \Theta} V(x,\bar{\theta}) = V(x,\theta) - \mu_{V,\Theta}(x,\theta) \leq V(x,\theta) - \delta$). These two properties play important roles in establishing GAS of the desired equilibrium set $\mathcal{A}$ for the overall closed-loop system.
	\end{rem}

	\begin{rem}
		Note that the design of the feedback term $\kappa(x,\theta)$ satisfying condition \eqref{eqn:prop_kappa} depends on the type of the switching variable $\theta$ and its flow dynamics $\varpi(x,\theta)$. If $\theta \in \mathbb{Z}$ denotes a logical index variable with $\varpi(x,\theta) \equiv 0$ as in \cite{mayhew2011synergistic}, condition \eqref{eqn:prop_kappa} can be reduced to
		$
		\langle \nabla_x V(x,\theta), f(x) + g(x)\kappa(x,\theta)\rangle  \leq 0.
		$ 
		If $\theta \in \mathbb{R}$ denotes a scalar variable with dynamics $ \varpi(x,\theta) =- k_\theta\nabla_\theta V(x,\theta)$ as in \cite{wang2022hybrid}, condition \eqref{eqn:prop_kappa} can be reduced to
		$
		\langle \nabla_x V(x,\theta), f(x) + g(x)\kappa(x,\theta)\rangle    - k_\theta\|\nabla_\theta V(x,\theta)\|^2 \leq 0.
		$ 
	\end{rem}

	\subsection{Hybrid Feedback with Smooth Control Input} \label{sec:hybrid_smooth}
	
	The switching variable $\theta$ with hybrid dynamics \eqref{eqn:theta_hybrid_dynamics} causes discontinuity in the feedback term $\kappa(x,\theta)$ during jumps, which is not desirable in practical applications. In order to remove the discontinuities in the hybrid feedback $\kappa(x,\theta)$, we assume that $\kappa(x,\theta)$ can be decomposed into some smooth functions of $x$ and a simplified function of $(x,\theta)$.
	More precisely, we assume that there exist functions $\varsigma: \mathcal{X}   \to \mathbb{R}^m$, $\Upsilon: \mathcal{X}  \to \mathbb{R}^{m\times s}$ and $\sigma: \mathcal{X} \times \mathbb{R}^r  \to \mathbb{R}^{s}$ such that
	\begin{align}
		\kappa(x,\theta) =
		\varsigma(x) + \Upsilon(x)\sigma(x,\theta) \label{eqn:kappa_decomp}.
	\end{align}
	We further assume that $\sigma$ satisfies the following assumption:
	\begin{assum}\label{assum:sigma_bound}
		There exists a constant $c_\kappa> 0 $ such that the function $\sigma$ in \eqref{eqn:kappa_decomp} satisfies $\max_{\bar{\theta}\in\Theta}  \|\sigma(x,\theta) - \sigma(x,\bar{\theta})\|^2  \leq  2c_\kappa$  for all $(x,\theta)\in \Psi_V \setminus\mathcal{A}$.
	\end{assum}

	\begin{rem}
		Note that the terms $\varsigma(x)$ and $\Upsilon(x)$ are smooth as they are independent of the switching variable $\theta$, and the function $\sigma(x,\theta)$ is smooth between the jumps of $\theta$. %, after each jump, satisfies Assumption \ref{assum:sigma_bound}.  
		By expressing  $\kappa(x,\theta)$ in the form of \eqref{eqn:kappa_decomp}, the upper bound condition in Assumption \ref{assum:sigma_bound} is only required on the simplified term $\sigma(x,\theta)$ instead of the entire feedback term $\kappa(x,\theta)$.  
		Note that the smoothing approach presented in \cite{mayhew2011synergistic} can be seen as a special case of our approach since the control input in \cite{mayhew2011synergistic} can rewritten as $\kappa(x,\theta) =  \Upsilon(x) \sigma(\theta):=\begin{bmatrix} \kappa(x,1), \dots, \kappa(x, L)\end{bmatrix}  e_{\theta}$ with $\theta \in \Theta:=\{1,2,\dots, L\}\subset \mathbb{N}$ and $e_\theta  \in \mathbb{R}^L$ denoting the $\theta$-th column of the identity matrix $I_L$. 
	\end{rem}
	
	In order to remove the discontinuities in the term $\sigma(x,\theta)$ in \eqref{eqn:kappa_decomp}, we introduce an auxiliary variable $\eta\in \mathbb{R}^s$ with continuous dynamics to smoothly track the non-smooth term $\sigma(x,\theta)$ under Assumption \ref{assum:sigma_bound}. More precisely, instead of directly implementing the non-smooth control input $u=\kappa(x,\theta)$, we consider the smooth control input $u =\bar{\kappa}(x,\eta):= \varsigma(x) + \Upsilon(x)\eta $ from \eqref{eqn:kappa_decomp} with $\dot{\eta} = u_s$ to be designed.  
	Let ${x}_s=(x,\eta)\in \mathcal{X}_s := \mathcal{X}\times \mathbb{R}^s$ denote the extended state.
	From \eqref{eqn:new_affine_system} and \eqref{eqn:kappa_decomp},
	one obtains the following modified nonlinear system:
	\begin{align}\label{eqn:smooth_affine_system}
		%\left.
		%\begin{array}{ll}
		\begin{pmatrix}
			\dot{x} \\
			\dot{\eta} \\
			\dot{\theta}
		\end{pmatrix}
		& = \underbrace{\begin{pmatrix}
				f(x) + g(x)\bar{\kappa}(x,\eta) \\%(\varsigma(x) + \Upsilon(x)\eta) \\
				0 \\
				\varpi(x,\theta)
		\end{pmatrix}}_{f_s(x_s,\theta)} +
		\underbrace{\begin{pmatrix}
				0 \\
				I_s \\
				0
		\end{pmatrix}}_{g_s(x_s,\theta)}u_s
	\end{align}
	where $(x_s,\theta)\in \mathcal{X}_s\times \mathbb{R}^r$ denotes the new state and $u_s  \in \mathbb{R}^s$ denotes the new control input. Then, the new objective consists in designing a new hybrid feedback $u_s = \kappa_s(x_s,\theta)$ for system \eqref{eqn:smooth_affine_system} such that the following equilibrium set of the overall closed-loop system is GAS:
	\begin{align}
		\mathcal{A}_s := \{(x_s,\theta)\in \mathcal{X}_s\times \mathbb{R}^r: (x,\theta)\in \mathcal{A}, \eta = \sigma(x,\theta)\}. \label{eqn:def_A_s}
	\end{align}
	
	Given a synergistic feedback quadruple $(V,\kappa,\varpi, \Theta)$ with $\kappa$ satisfying the form of \eqref{eqn:kappa_decomp}, we consider the following real-valued functions $V_s: \mathcal{X}_s\times \mathbb{R}^r  \to \mathbb{R}_{\geq 0}$ and $\kappa_s: \mathcal{X}_s \times \mathbb{R}^r \to \mathbb{R}^s$:
	\begin{align}
		V_s(x_s,\theta) & =  V(x,\theta) + \frac{\gamma_s}{2}\|\eta - \sigma(x,\theta)\|^2 \label{eqn:def_V_s}\\
		\kappa_s(x_s,\theta)
		& =  -  k_\eta(\eta-\sigma(x,\theta))  +   \mathcal{D}_t{\sigma}(x,\theta)   \nonumber\\
		& \qquad  \qquad -  \frac{1}{\gamma_s} \Upsilon(x)\T g(x)\T \nabla_x V(x,\theta)   \label{eqn:def_kappa_s}
	\end{align}
	with constant scalars $k_\eta, \gamma_s>0$ to be designed, and 
	\begin{multline*}
		\mathcal{D}_t{\sigma}(x,\theta):= \mathcal{D}_x \sigma (x,\theta) \left( f(x) + g(x)\bar{\kappa}(x,\eta) \right)  \\
		+ \mathcal{D}_\theta \sigma (x,\theta)\varpi(x,\theta)
	\end{multline*}
	with $\mathcal{D}_x \sigma(x,\theta)$ and $\mathcal{D}_\theta \sigma(x,\theta)$ denoting the Jacobians of $ \sigma$ relative to $x$ and $\theta$, respectively, \ie, $\mathcal{D}_t{\sigma}(x,\theta)= \frac{d}{dt}{\sigma}(x,\theta)$ between the jumps of $\theta$.  
	From the definition of $V_s$ in \eqref{eqn:def_V_s} and the properties of the synergistic feedback quadruple $(V,\kappa,\varpi,\Theta)$ in Definition \ref{defn:synergistic_feedback}, one can verify that $V_s$ is positive definite with respect to the set $\mathcal{A}_s$.

	\begin{prop}\label{prop:SLFS}
		Consider  the real-valued functions $V_s,\kappa_s$ defined in \eqref{eqn:def_V_s}-\eqref{eqn:def_kappa_s} with $(V,\kappa,\varpi,\Theta)$ being a synergistic feedback quadruple relative to $\mathcal{A}$, with gap exceeding $\delta>0$, for system \eqref{eqn:new_affine_system}. Suppose that Assumption \ref{assum:sigma_bound} holds, and
		choose $0<\gamma_s < \frac{\delta}{c_k}$ and $k_\eta>0$. Then, $(V_s,\kappa_s,\varpi,\Theta)$ is a synergistic feedback quadruple relative to the set $\mathcal{A}_s$ in \eqref{eqn:def_A_s}, with gap exceeding $\delta_s\in (0,\delta- \gamma_s c_k ]$, for system \eqref{eqn:smooth_affine_system}.
	\end{prop}
	\begin{proof}
		See Appendix \ref{sec:SLFS}.
	\end{proof}
	
	Given the synergistic feedback quadruple $(V_s,\kappa_s,\varpi,\Theta)$ as per Proposition \ref{prop:SLFS}, we propose the following modified hybrid feedback for system \eqref{eqn:affine_system}:
	\begin{align}
		\underbrace{
			\begin{array}{ll}
				u            & = \varsigma(x) + \Upsilon(x)\eta \\
				\dot{\eta}   & =  \kappa_s(x_s,\theta) \\
				\dot{\theta} & = \varpi(x,\theta)
			\end{array}
		}_{(x_s,\theta)\in \mathcal{F}_s}~
		\underbrace{
			\begin{array}{ll}
				&\\
				\eta^+   & = \eta\\
				\theta^+ & \in G_s(x_s,\theta)
			\end{array}
		}_{(x_s,\theta)\in  \mathcal{J}_s} \label{eqn:hybrid_feedback2}
	\end{align}
	where $\kappa_s$ is defined in \eqref{eqn:def_kappa_s} and  $G_s, \mathcal{F}_s, \mathcal{J}_s$ are designed as follows:
	\begin{subequations}\label{eqn:hybrid_feedback2_design}
		\begin{align}
			%u_s           & =\kappa_s(x_s,\theta)\\
			%\varpi(x_s,\theta) & =- k_\theta\nabla_\theta V(x,\theta)  \label{eqn:def_f_theta}\\
			G_s(x_s,\theta) & = \{\bar{\theta} \in \Theta: \mu_{V_s,\Theta}(x_s,\bar{\theta})= 0 \} \label{eqn:def_G_theta_s} \\
			\mathcal{F}_s   & = \{(x_s,\theta)\in \mathcal{X}_s\times \mathbb{R}^r: \mu_{V_s,\Theta}(x_s,\theta) \leq \delta_s\} \label{eqn:def_Fs_set}\\
			\mathcal{J}_s   & = \{(x_s,\theta)\in \mathcal{X}_s\times \mathbb{R}^r: \mu_{V_s,\Theta}(x_s,\theta) \geq \delta_s\}  \label{eqn:def_Js_set}
		\end{align}
	\end{subequations}
	with %$\mu_{V_s,\Theta}: \mathcal{X}_s \times \mathbb{R}^r \to \mathbb{R}$  defined as 
	$$\mu_{V_s,\Theta}(x_s,\theta) :=   V_s(x_s,\theta) - \min_{\bar{\theta}\in \Theta} V_s(x_s,\bar{\theta})$$ 
	and some $\delta_s \in (0,\delta- \gamma_s c_k ]$ chosen as per Proposition \ref{prop:SLFS}.
	In view of \eqref{eqn:smooth_affine_system}, \eqref{eqn:def_kappa_s} and \eqref{eqn:hybrid_feedback2}, one obtains the following closed-loop system:
	\begin{subequations}\label{eqn:hybrid_closed_loop2}
		\begin{align} 
			\begin{pmatrix}
				\dot{x} \\
				\dot{\eta} \\
				\dot{\theta}
			\end{pmatrix} &                                                              %= F_o(x,\theta) :
			= \begin{pmatrix}
				f(x) + g(x)\bar{\kappa}(x,\eta)\\%f_s(x_s) + g_s(x_s)\kappa_s(x_s,\theta) \\
				\kappa_s(x_s,\theta)\\
				\varpi(x,\theta)
			\end{pmatrix}, (x_s,\theta)\in \mathcal{F}_s \label{eqn:hybrid_closed_F_s} \\
			\begin{pmatrix}
				x^+ \\
				\eta^+ \\
				\theta^+
			\end{pmatrix} & \in                                                          %G_o(x,\theta) := 
			\begin{pmatrix}
				x \\
				\eta \\
				G_s(x_s,\theta)
			\end{pmatrix},\qquad \qquad  (x_s,\theta)\in \mathcal{J}_s \label{eqn:hybrid_closed_G_s}
		\end{align}
	\end{subequations}
	with $\kappa_s$ defined in \eqref{eqn:def_kappa_s}, and $G_s, \mathcal{F}_s, \mathcal{J}_s$ defined in \eqref{eqn:hybrid_feedback2_design}.  
	Similar to Lemma \ref{lem:HBC}, one can show that the hybrid closed-loop system \eqref{eqn:hybrid_closed_loop2} satisfies the hybrid basic conditions \cite[Assumption 6.5]{goebel2012hybrid}. Then, one can state one of our main results:
	\begin{thm}\label{thm:thm2}
		Consider the hybrid closed-loop system \eqref{eqn:hybrid_closed_loop2} with the synergistic feedback quadruple $(V_s,\kappa_s,\varpi,\Theta)$ with gap exceeding $\delta_s>0$   as per Proposition \ref{prop:SLFS}.
		Let Assumption \ref{assum:sigma_bound} hold.
		Then, the compact set $\mathcal{A}_s$ in \eqref{eqn:def_A_s} is GAS for system \eqref{eqn:hybrid_closed_loop2} if $f(x) + g(x)(\varsigma(x) + \Upsilon(x)\eta) \in \mathsf{T}_\mathcal{X}(x)$  for all $x_s=(x,\eta)\in \mathcal{X}_s$.
	\end{thm}
	
	From Proposition \ref{prop:SLFS}, one can show that $(V_s,\kappa_s,\varpi,\Theta)$ is a synergistic feedback quadruple relative to $\mathcal{A}_s$, with gap exceeding $\delta_s$, for system \eqref{eqn:smooth_affine_system}. Therefore, the proof of Theorem \ref{thm:thm2} is omitted as it follows closely the proof of Theorem \ref{thm:thm1} by applying the results of Proposition \ref{prop:SLFS}.
	
	\subsection{Integrator Backstepping}
	In this subsection, we extend the dynamics in \eqref{eqn:smooth_affine_system} to include the control input $u$ as an additional controller state and obtain the following extended affine nonlinear system:
	\begin{align}\label{eqn:backstepping_affine_system}
		\begin{pmatrix}
			\dot{x} \\
			\dot{\eta} \\
			\dot{u} \\
			\dot{\theta}
		\end{pmatrix} =  \underbrace{\begin{pmatrix}
				f(x) + g(x)u\\
				\kappa_s(x_s,\theta) \\
				0\\
				\varpi(x,\theta)
		\end{pmatrix}}_{f_b(x_b,\theta)} +
		\underbrace{\begin{pmatrix}
				0 \\
				0 \\
				I_m \\
				0
		\end{pmatrix}}_{g_b(x_b,\theta)} u_b
	\end{align}
	where $x_b:=(x_s,u) \in \mathcal{X}_b:=\mathcal{X}_s \times \mathbb{R}^m$ denotes the new extended state, and $u_b \in \mathbb{R}^m$ denotes the new control input to be designed. 
	The new goal consists in designing a new hybrid feedback $u_b = \kappa_b(x_b,\theta)$ for system \eqref{eqn:backstepping_affine_system} such that the following equilibrium set of the overall closed-loop system is GAS:
	\begin{align} \label{eqn:def_A_b}
		\mathcal{A}_b : = \{(x_b,\theta)\in \mathcal{X}_b \times \mathbb{R}^r: (x_s,\theta) \in \mathcal{A}_s,  u = \bar{\kappa}(x,\eta)\}
	\end{align}
	with $\mathcal{A}_s$ defined in \eqref{eqn:def_A_s} and $\bar{\kappa}(x,\eta) = \varsigma(x) + \Upsilon(x)\eta$.
	Given a synergistic feedback quadruple $(V_s,\kappa_s,\varpi,\Theta)$, we consider the following real-valued functions $V_b: \mathcal{X}_b\times \mathbb{R}^r  \to \mathbb{R}_{\geq 0}$ and $\kappa_b: \mathcal{X}_b \times \mathbb{R}^r \to \mathbb{R}^m$:
	\begin{align}
		V_b(x_b,\theta) & =      %V(x,\theta) + \frac{\gamma_s}{2}\|\eta - \sigma(x,\theta)\|^2 \\
		V_s(x_s, \theta)+  \frac{\gamma_b}{2} \|u - \bar{\kappa}(x,\eta)\|^2 \label{eqn:def_V_b} \\
		\kappa_b(x_b,\theta)
		& =  -  k_b(u- \bar{\kappa}(x,\eta)) + \mathcal{D}_t \bar{\kappa}(x,\eta)  \nonumber\\
		& \qquad \qquad   -  \frac{1}{\gamma_b} g(x)\T \nabla_x V(x,\theta)   \label{eqn:def_kappa_b}
	\end{align}
	where $\gamma_b, k_b >0$ are constant scalars to be designed, and %$\mathcal{D}_t \bar{\kappa}(x,\eta) : = \frac{d}{dt}{\bar{\kappa}}(x,\eta)=  (  \mathcal{D}_x{\varsigma}(x) + \sum_{i=1}^{s} \eta_i \mathcal{D}_x \Upsilon_i(x) ) (f(x)+g(x)u) + \Upsilon(x)\kappa_s(x_s,\theta)  $  
	\begin{multline*}
		\mathcal{D}_t \bar{\kappa}(x,\eta)
		:=  \Upsilon(x)\kappa_s(x_s,\theta) \\
		+ \left(  \mathcal{D}_x{\varsigma}(x) + \sum_{i=1}^{s} \eta_i \mathcal{D}_x \Upsilon_i(x) \right) (f(x)+g(x)u)
	\end{multline*}
	with $\eta=[\eta_i,\dots,\eta_s]\T$,  $ \Upsilon(x) =[\Upsilon_1(x),\dots, \Upsilon_s(x)]$ such that $\Upsilon(x)\eta = \sum_{i=1}^{s}\Upsilon_i(x)\eta_i$, and $\mathcal{D}_x{\varsigma}(x), \mathcal{D}_x \Upsilon_i(x)$ denoting the Jacobian matrices of ${\varsigma}(x)$ and $\Upsilon_i(x)$ related to $x$, respectively.
	
	%The following proposition establishes that $(V_b,\kappa_b,\varpi,\Theta)$ is a synergistic feedback quadruple.
	\begin{prop}\label{prop:SLFS_backstepping}
		Consider the real-valued functions $V_b,\kappa_b$ defined in \eqref{eqn:def_V_b}-\eqref{eqn:def_kappa_b} with $(V,\kappa,\varpi,\Theta)$ being a synergistic feedback quadruple relative to $\mathcal{A}$, with gap exceeding $\delta>0$, for system \eqref{eqn:new_affine_system}. Suppose that Assumption \ref{assum:sigma_bound} holds, and choose $0<\gamma_s < \frac{\delta}{c_k}, \gamma_b>0$. Then, $(V_b,\kappa_b,\varpi,\Theta)$ is a synergistic feedback quadruple relative to the set $\mathcal{A}_b$ in \eqref{eqn:def_A_b}, with gap exceeding $\delta_b\in (0,\delta- \gamma_s c_k ]$, for system \eqref{eqn:backstepping_affine_system}.
	\end{prop}
	\begin{proof}
		See Appendix \ref{sec:SLFS_backstepping}.
	\end{proof}
	
	Given the synergistic feedback quadruple $(V_b,\kappa_b,\varpi,\Theta)$ as per Proposition \ref{prop:SLFS_backstepping}, we propose the following modified hybrid feedback for system \eqref{eqn:backstepping_affine_system}:
	\begin{align}
		\underbrace{
			\begin{array}{ll}
				\dot{u}      & = \kappa_b(x_b,\theta)\\
				\dot{\eta}   & =  \kappa_s(x_s,\theta) \\
				\dot{\theta} & = \varpi(x,\theta)
			\end{array}
		}_{(x_b,\theta)\in \mathcal{F}_b}~
		\underbrace{
			\begin{array}{ll}
				u^+      & = u\\
				\eta^+   & = \eta\\
				\theta^+ & \in G_b(x_b,\theta)
			\end{array}
		}_{(x_b,\theta)\in  \mathcal{J}_b} \label{eqn:hybrid_feedback3}
	\end{align}
	where $\kappa_s$ and $\kappa_b$ are defined in \eqref{eqn:def_kappa_s} and \eqref{eqn:def_kappa_b}, respectively, and $G_b, \mathcal{F}_b, \mathcal{J}_b$  are designed as follows:
	\begin{subequations}\label{eqn:hybrid_feedback3_design}
		\begin{align}
			%u_s           & =\kappa_s(x_s,\theta)\\
			%\varpi(x_s,\theta) & =- k_\theta\nabla_\theta V(x,\theta)  \label{eqn:def_f_theta}\\
			G_b(x_b,\theta) & = \{\bar{\theta} \in \Theta: \mu_{V_b,\Theta}(x_b,\bar{\theta})= 0 \} \label{eqn:def_G_theta_b} \\
			\mathcal{F}_b   & = \{(x_b,\theta)\in \mathcal{X}_b\times \mathbb{R}^r: \mu_{V_b,\Theta}(x_b,\theta) \leq \delta_b\} \label{eqn:def_Fb_set}\\
			\mathcal{J}_b   & = \{(x_b,\theta)\in \mathcal{X}_b\times \mathbb{R}^r: \mu_{V_b,\Theta}(x_b,\theta) \geq \delta_b\}  \label{eqn:def_Jb_set}
		\end{align}
	\end{subequations}
	with %\mu_{V_b,\Theta}: \mathcal{X}_b \times \mathbb{R}^r \to \mathbb{R}$  defined as
	$$
	\mu_{V_b,\Theta}(x_b,\theta)  :=   V_b(x_b,\theta) - \min_{\bar{\theta}\in \Theta} V_b(x_b,\bar{\theta})    
	$$ 
	and some $\delta_b \in (0,\delta- \gamma_s c_k ]$ chosen as per Proposition \ref{prop:SLFS_backstepping}.
	In view of \eqref{eqn:backstepping_affine_system}, \eqref{eqn:def_kappa_b} and \eqref{eqn:hybrid_feedback3}, one obtains the following hybrid closed-loop system:
	\begin{subequations}\label{eqn:hybrid_closed_loop3}
		\begin{align}
			\begin{pmatrix}
				\dot{x} \\
				\dot{\eta} \\
				\dot{u} \\
				\dot{\theta}
			\end{pmatrix} &
			= \begin{pmatrix}
				f(x) + g(x)u\\
				\kappa_s(x_s,\theta) \\
				\kappa_b(x_b,\theta) \\
				\varpi(x,\theta)
			\end{pmatrix}, \quad (x_b,\theta)\in \mathcal{F}_b\label{eqn:hybrid_closed_F_b} \\
			\begin{pmatrix}
				x^+ \\
				\eta^+ \\
				u^+ \\
				\theta^+
			\end{pmatrix} & \in                                                         %G_o(x,\theta) := 
			\begin{pmatrix}
				x \\
				\eta \\
				u \\
				G_b(x_b,\theta)
			\end{pmatrix},\quad  ~~~    (x_b,\theta)\in \mathcal{J}_b \label{eqn:hybrid_closed_G_b}
		\end{align}
	\end{subequations}
	with  $\kappa_s$ and $\kappa_b$ defined in \eqref{eqn:def_kappa_s} and \eqref{eqn:def_kappa_b}, respectively,  and $G_b, \mathcal{F}_b, \mathcal{J}_b$ defined in \eqref{eqn:hybrid_feedback3_design}. 
	Similar to Lemma \ref{lem:HBC}, one can also show that the hybrid closed-loop system \eqref{eqn:hybrid_closed_loop3} satisfies the hybrid basic conditions \cite[Assumption 6.5]{goebel2012hybrid}. The stability properties of the closed-loop system \eqref{eqn:hybrid_closed_loop3} are given in the following theorem:
	\begin{thm}\label{thm:thm3}
		Consider the closed-loop system \eqref{eqn:hybrid_closed_loop3} with the synergistic feedback quadruple $(V_b,\kappa_b,\varpi,\Theta)$ with gap exceeding $\delta_b>0$ as per Proposition \ref{prop:SLFS_backstepping}.
		Let Assumption \ref{assum:sigma_bound} hold.
		Then, the compact set $\mathcal{A}_b$ in \eqref{eqn:def_A_b} is GAS for system \eqref{eqn:hybrid_closed_loop3} if $f(x) + g(x)u \in \mathsf{T}_\mathcal{X}(x)$  for all $x_b=(x,\eta, u)\in \mathcal{X}_b$.
	\end{thm}
	
	The proof of  Theorem \ref{thm:thm3} is omitted as it follows closely the proof of  Theorem \ref{thm:thm1}, using the fact that $(V_b,\kappa_b,\varpi,\Theta)$ is a synergistic feedback quadruple relative to the set $\mathcal{A}_b$ for system \eqref{eqn:backstepping_affine_system} with gap exceeding $\delta_b$, as per Proposition \ref{prop:SLFS_backstepping}.

	\section{Extension to Safe Navigation with Global Obstacle Avoidance}\label{sec:navigation}
	In this section, we apply the synergistic hybrid feedback approach developed in the previous section to the navigation problem of a planar robot with single integrator dynamics, relying on a modified navigation function. We consider a point mass robot navigating in the workspace $\mathbb{R}^2$ and a circular obstacle centered at $p_{o} \in \mathbb{R}^2$ with radius  $r_{o}>0$, \ie, $\mathcal{O}:= \{z\in \mathbb{R}^2: \|z-p_{o}\| \leq r_{o}\}$.  
	The free space $\mathcal{X}_p$ for the robot is given as the closed set 
	$$\mathcal{X}_p:= \overline{\mathbb{R}^2\setminus (\mathcal{O} \oplus \varepsilon  \mathbb{B})}= \{p\in \mathbb{R}^2: \|p-p_o\| \geq \bar{r}_{o} := r_o + \varepsilon \}$$
	with a sufficiently small safety margin $\varepsilon>0$ defining the minimum safety distance to the obstacle.
	We consider the following single integrator dynamics:
	\begin{align}
		\dot{p} & = u \label{eqn:system_robot}  
	\end{align} 
	where $p \in \mathcal{X}_p$ denotes the position of the robot, $u\in \mathbb{R}^2$ is the velocity input.  
	The objective is to design hybrid feedback for system \eqref{eqn:system_robot} such that, for any initial condition $p(0)\in \mathcal{X}_p$, the destination point $p_d \in \mathcal{X}_p$ is GAS and collision-free navigation is guaranteed (\ie, $p(t)\in  \mathcal{X}_p$ for all $t\geq 0$). 
	\subsection{Synergistic Navigation Function}
	%Let $d_o(p) : =  \|p-p_{o}\|  -r_{o}$ denote the distance from the position $p$ to the obstacle $\mathcal{O}$. Then, one has $d_o(p) \geq \varepsilon$ for all $p \in \mathcal{X}_p$.
	Define the distance from the robot position $p$ to the obstacle $\mathcal{O}$ as
	\begin{align*}
		d_o(p) & : =  \|p-p_{o}\|  -r_{o}.  
	\end{align*}  
	Then, one can verify that  $d_o(p) \geq \varepsilon$ for all $p \in \mathcal{X}_p$.
	% One can verify that $d_o(p) \geq \varepsilon$ for all $p \in \mathcal{X}_p$.
	Consider the following smooth navigation function $V_{nav}: \mathcal{X}_p \to \mathbb{R}_{\geq 0}$ as:
	\begin{equation}\label{eqn:def_navigation}
		V_{nav}(p) = \frac{1}{2}\|p - p_d\|^2 + \varrho \phi(d_o(p))  
	\end{equation}
	where $\varrho$ is a strictly positive scalar and the barrier function $\phi: \mathbb{R}_{\geq 0}\to \mathbb{R}_{\geq 0}$ is designed as in \cite{sanfelice2006robust}  
	\begin{equation}\label{eqn:Phi_def}
		\phi(z) = \begin{cases} 
			(z-r_s)^2\ln(\frac{r_s}{z}) & z\in (0,r_s] \\
			0                           & z > r_s
		\end{cases}
	\end{equation}
	where $r_s\in (\varepsilon, r_d)$ determines the distance when the term $\phi(d_o(p))$ varnishes, with $r_d:= d_o(p_d)$ denoting the distance between the obstacle and the destination location. %Choosing $r_s < r_d $  ensures  that $d_o(p_d)>r_s$ and $\phi(d_o(p_d))=0$. 
	The gradient of the barrier function $\phi$ is given by
	\begin{align}
		\nabla_z\phi(z) & = 
		\begin{cases} 
			2(z-r_s) \ln(\frac{r_s}{z}  ) - \frac{ (z-r_s)^2 }{z} & z\in (0,r_s] \\
			0                                                     & z >  r_s
		\end{cases} .\label{eqn:DPhi_def}
	\end{align}
	% \end{rem} 
From \eqref{eqn:Phi_def} and \eqref{eqn:DPhi_def}, the function $\phi$ satisfies the following properties:
\begin{itemize}
	\item [1)] $\phi(z)$ is  continuously differentiable;
	\item [2)] $\phi(z)$ is strictly decreasing on $(0,r_s]$ with $\lim_{z\to 0}\phi(z) = + \infty$ and $\phi(z) = 0$ for all $z\geq r_s$;
	\item [3)] $\nabla_z \phi(z)  \leq 0$ for all $z\in (0,r_s]$ and $\nabla_z \phi(z) = 0$ for all $z \geq r_s$.
\end{itemize}
Choosing $r_s\in (\varepsilon,r_d)$ ensures that $d_o(p)>r_s$ and $\phi(d_o(p))=0$ when the robot arrives at the destination location, \ie, $p=p_d$. One can verify that the navigation function $V_{nav}$ in \eqref{eqn:def_navigation} is positive definite with respect to the destination $p_d$, \ie,  $V_{nav}(p)\geq 0$ for all $p\in \mathcal{X}_p$, $V_{nav}(p_d)=0$ and $V_{nav}(p) \to \infty$ as $p$ approaches infinity or the boundary of the obstacle. 
Moreover, the gradient and the set of critical points of $V_{nav}$ in \eqref{eqn:def_navigation} are given by
\begin{align}
	\nabla_{p} V_{nav}(p)  &= p - p_d + \varrho \nabla_{d_o(p)} \phi(d_o(p)) \frac{p-p_o}{\|p-p_o\|} \label{eqn:gradient_V_nav} \\
	C_{V_{nav}} & =\{p\in \mathcal{X}_p:\nabla_{p} V_{nav}(p) = 0\}. \label{eqn:critical_V_nav}
\end{align}
% \begin{rem}
	Note that the vector $p - p_d$ in \eqref{eqn:gradient_V_nav} represents the attraction field that steers the robot towards the destination location, and the vector $\nabla_{d_o(p)} \phi(d_o(p)) \frac{p-p_o}{\|p-p_o\|}$ represents the repulsive field that repels the robot away from the obstacle along the direction $p_o-p$. The parameter $r_s$ is designed to trigger the repulsive field based on the distance between the robot and the obstacle, and the weights of these two fields can be adjusted by tuning the scalar $\varrho$.  
	It follows from \eqref{eqn:gradient_V_nav} and \eqref{eqn:critical_V_nav} that $p_d \in C_{V_{nav}}$ and there exists another critical point $p^* \in C_{V_{nav}} \setminus \{p_d\}$ (on the side of the obstacle opposite to the target) satisfying $\nabla_{p} V_{nav}(p^*)=0$ %$p^*-p_d = -\lambda(p^*) (p^*-p_o)$ with some constant $\lambda(p^*) = \nabla_{d(p^*)} \phi(d(p^*)) \frac{1}{\|p^*-p_o\|} >0$ 
	(see, for instance, the example presented in Fig. \ref{fig:example1}). Due to the existence of this undesired critical point, it is impossible to achieve global navigation using the time-invariant gradient-based feedback of the smooth navigation function $V_{nav}$ (see, for instance, \cite{koditschek1990robot}).  
	% \end{rem} 

\begin{figure}[!ht]
	\centering
	\includegraphics[width=0.8\linewidth]{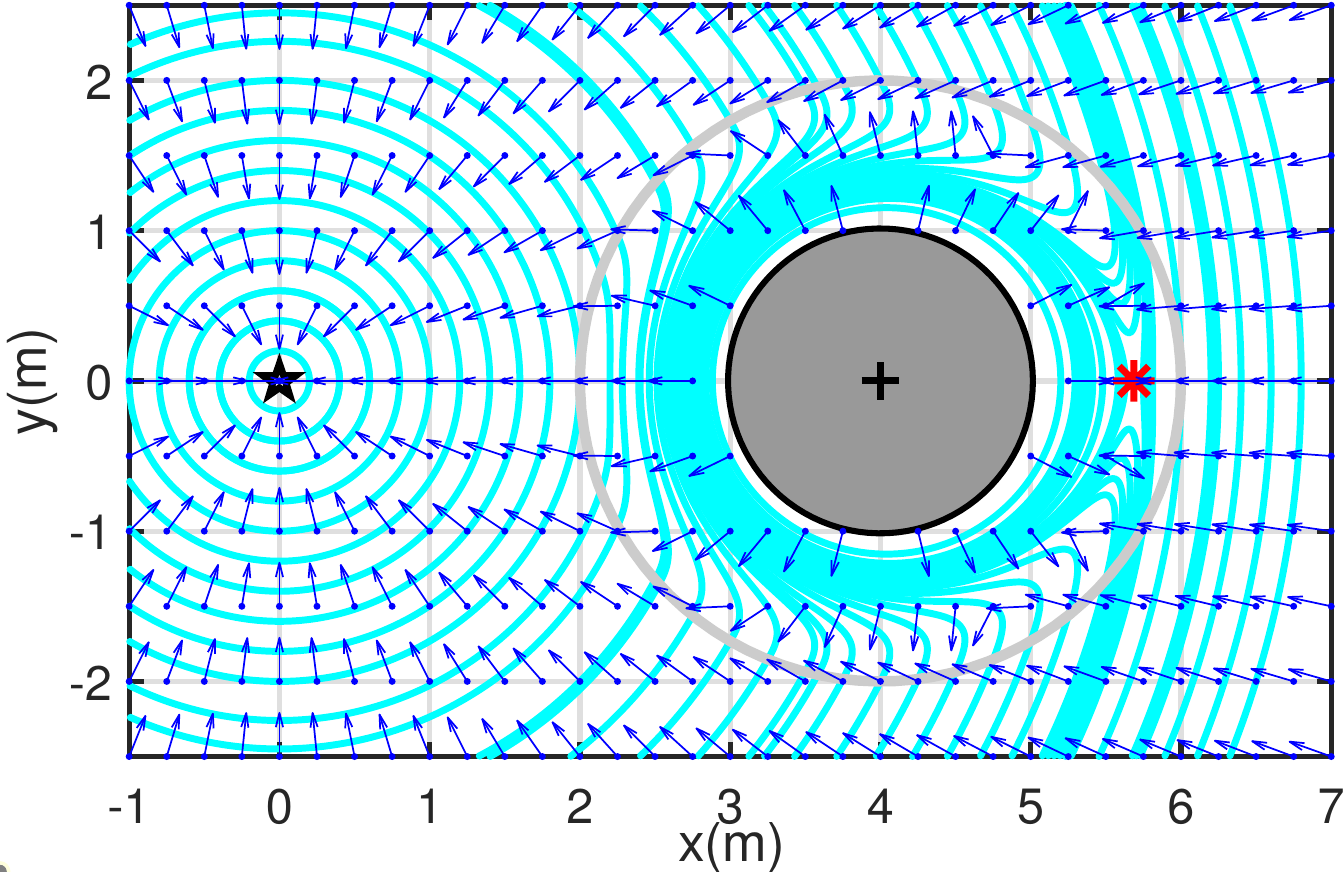}
	\caption{Example of the smooth navigation function $V_{nav}$ defined in \eqref{eqn:def_navigation} with $r_s = 1$ (gray circle) and $\varrho = 16$, and a circular obstacle centered at $p_o = (4,0)$ (black $+$) with radius $r_o=1$. The destination is chosen at $p_d=(0,0)$ (black $\star$), and the undesired (saddle) critical point is located at $p^*=(5.6865,0)$ (red $\ast$). The blue arrows and the light blue lines represent the direction of the gradients and the contour of the navigation function, respectively.} 
	\label{fig:example1}
\end{figure}

To apply the synergistic hybrid feedback, we introduce the transformation function  $\mathcal{T}:\mathcal{X}_p \times \mathbb{R} \to \mathcal{X}_p$ defined as
\begin{equation}
	\mathcal{T}(p,\theta):=   p_o + \mathcal{R}(\theta) (p-p_o)    \label{eqn:def_mathcal_T}
\end{equation}
where $\theta\in \mathbb{R}$ is a switching variable and 
the rotation matrix $\mathcal{R}(\theta) := \exp(\theta \Delta) \in SO(2)$\footnote{$SO(2):= \{R\in \mathbb{R}^{2\times 2}: R\T R = R R\T  = I_2, \det{R}=+ 1\}$ denotes the 2-dimensional Special Orthogonal group.} is given by 
\begin{align}
	\mathcal{R}(\theta) &  = \begin{bmatrix}
		\cos(\theta) & -\sin(\theta) \\
		\sin(\theta) & \cos(\theta)
	\end{bmatrix} , \quad   \Delta = \begin{bmatrix}
		0 & -1 \\
		1 & 0
	\end{bmatrix}.   \label{eqn:def_R}
\end{align}
\begin{figure}[!ht]
	\centering
	\includegraphics[width=0.85\linewidth]{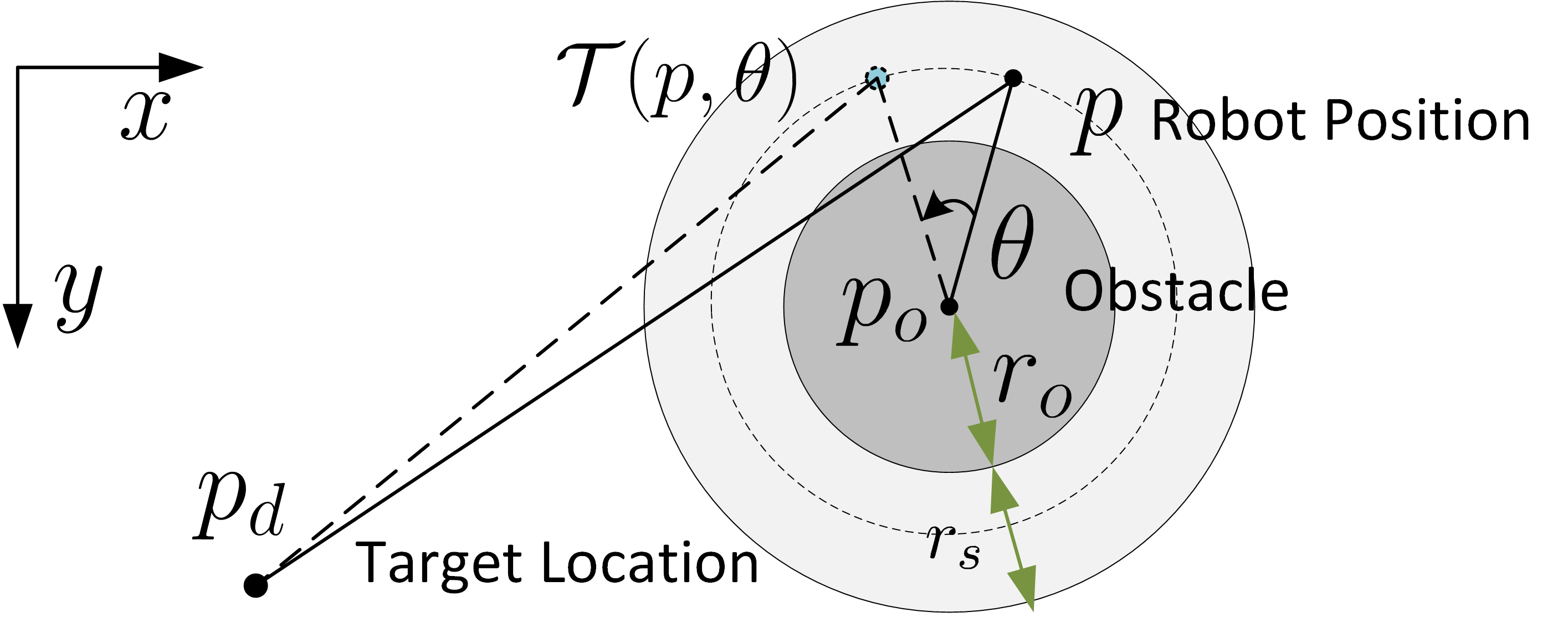}
	\caption{Geometric representation of the transformation function $\mathcal{T}$.}
	\label{fig:diagram21}
\end{figure} 

Consider the following modified navigation function $\mathcal{V}_{nav}: \mathcal{X}_p \times \mathbb{R} \to \mathbb{R}_{\geq 0}$:
\begin{align}
	\mathcal{V}_{nav}(p,\theta) & = V_{nav}(\mathcal{T}(p,\theta)) + \frac{\gamma_\theta}{2} \theta^2
	\label{eqn:def_synergistic_navigation}
\end{align}
where $\gamma_\theta>0$ and $V_{nav}$ is the smooth navigation function defined in \eqref{eqn:def_navigation}.  
From \eqref{eqn:def_mathcal_T}, the transformation function $\mathcal{T}$ satisfies: 1) $\mathcal{T}(p,0) = p$ for all $p\in \mathcal{X}_p$ since $\mathcal{R}(0) = I_2$; 2) its Jacobian relative to $p$ (\ie, $\mathcal{D}_p\mathcal{T}(p,\theta) = \mathcal{R}(\theta)$) is non-singular for all $(p,\theta) \in \mathcal{X}_p \times \mathbb{R}$ since $\mathcal{R}(\theta) \in SO(2)$.  
The main idea behind introducing such a transformation function $\mathcal{T}$ is to (virtually) transform the robot position from $p$ to a new position $p'=\mathcal{T}(p,\theta)$ by rotating the vector $p-p_o$ by some angle $\theta$ such that the gradient of the modified navigation function $\mathcal{V}_{nav}$, denoted by $\nabla_{p} \mathcal{V}_{nav}(p,\theta)$, becomes non-zero after switching $\theta$ at the undesired critical point.  
From \eqref{eqn:def_mathcal_T}, the distance from the robot to the obstacle remains the same, \ie, $d_o(\mathcal{T}(p,\theta)) = d_o(p)$. 
Hence, from \eqref{eqn:def_navigation} and \eqref{eqn:def_synergistic_navigation}, the modified navigation function $\mathcal{V}_{nav}$ can be rewritten as  
\begin{align}
	\mathcal{V}_{nav}(p,\theta)  
	& = \frac{1}{2}\|\mathcal{T}(p,\theta) - p_d\|^2 + \varrho \phi(d_o(p)) + \frac{\gamma_\theta}{2} \theta^2	\label{eqn:def_synergistic_navigation2}
\end{align}
with some $\gamma_\theta>0$ to be designed. In view of the facts $\mathcal{T}(p,0)=p$ and $d_o(\mathcal{T}(p,\theta)) = d_o(p)$, one can verify that $\phi(d_o(\mathcal{T}(p,\theta))) = \phi(d_o(p))=0$ for all $d_o(p)>r_s$, and $\mathcal{V}_{nav}$ is positive definite with respect to the equilibrium point $\mathcal{A}_p:=(p_d,0)$. Moreover, one can show that $\mathcal{V}_{nav}$ has a unique minimum at $\mathcal{A}_p$, and $\mathcal{V}_{nav}(p,\theta)\to \infty$ as $p$ approaches the obstacle or infinity. The following lemma provides the explicit gradients $\nabla_{p} \mathcal{V}_{nav}(p,\theta)$ and $\nabla_{\theta} \mathcal{V}_{nav}(p,\theta)$, as well as the set of critical points $C_{\mathcal{V}_{nav}}$ of the navigation function $\mathcal{V}_{nav}$ in \eqref{eqn:def_synergistic_navigation}.

\begin{lem}\label{lem:grad_V_nav2}
	From the navigation function $\mathcal{V}_{nav}$ in \eqref{eqn:def_synergistic_navigation2}, one has
	\begin{align}
		\nabla_{p} \mathcal{V}_{nav}(p,\theta)      & =  \nabla_{p} V_{nav}(p)  + (I_2-\mathcal{R}(\theta))\T (p_d-p_o) \label{eqn:def_grad_V'_p}\\
		\nabla_{\theta} \mathcal{V}_{nav}(p,\theta) & = 	\gamma_\theta \theta - (p-p_o)\T \Delta \mathcal{R}(\theta)\T (p_o - p_d) \label{eqn:def_grad_V'_zeta} \\
		C_{\mathcal{V}_{nav}} &= C_{V_{nav}}\times \{0\}  \label{eqn:def_V'_critical_set}
		%\{(p,\theta) \in \mathcal{X}_p \times \mathbb{R}: p\in C_{V_{nav}}, \theta = 0\}
	\end{align}
	with $\nabla_{p} V_{nav}(p)$ and $C_{V_{nav}}$ defined in \eqref{eqn:gradient_V_nav} and \eqref{eqn:critical_V_nav}, respectively. 
\end{lem}
\begin{proof}
	See Appendix \ref{sec:grad_V_nav2}.
\end{proof}
Note that the main difference between the gradients $\nabla_{p}V_{nav}(p)$ in \eqref{eqn:gradient_V_nav} and $\nabla_{p}\mathcal{V}_{nav}(p,\theta)$ in \eqref{eqn:def_grad_V'_p} is the additional term  $(I_2-\mathcal{R}(\theta))\T (p_o-p_d)$, which generates an additional vector field non-collinear to the vector $p_o-p_d$ if $\mathcal{R}(\theta)\neq I_2$. In addition, when $\nabla_{p} \mathcal{V}_{nav}(p,\theta)$ vanishes at an undesired critical point $(p^*, 0) \in C_{\mathcal{V}_{nav}} \setminus \mathcal{A}_p$, the switching variable $\theta$ can be set to some value from the pre-design finite nonempty set $\Theta$ to steer the state, after a jump, away from this undesired critical point, and generate an additional non-zero vector field in the gradient. This additional term in \eqref{eqn:def_grad_V'_p}, together with an appropriate feedback design, will steer the state away from all the undesired critical points and converge to the desired one $\mathcal{A}_p=(p_d, 0)$.

\begin{defn}
	Let $\Theta\subset \mathbb{R}$ be a finite nonempty set. A real-valued navigation function $\mathcal{V}_{nav}: \mathcal{X}_p \times \mathbb{R} \to \mathbb{R}_{\geq 0}$ is called a \textit{synergistic navigation function} relative to $\mathcal{A}_{p}$ with gap exceeding $\delta_{\mathcal{V}}>0$ if %there exists $\delta_{\mathcal{V}}$ such that    
	\begin{equation}
		\mathcal{V}_{nav}(p,\theta) - \min_{\bar{\theta}\in\Theta}\mathcal{V}_{nav}(p,\bar{\theta}) > \delta_{\mathcal{V}}
	\end{equation}
	for all $(p,\theta) \in C_{\mathcal{V}_{nav}}\setminus \mathcal{A}_{p}$.
\end{defn}
\begin{lem}\label{lem:V_nav_synergistic}
	Consider the finite nonempty  set $\Theta=\{|\theta|_i\in (0,\pi),i=1,\dots,L\}$. 
	The real-valued function $\mathcal{V}_{nav}$  defined in \eqref{eqn:def_synergistic_navigation2} is a {synergistic navigation function} on $\mathcal{X}_{p} \times \mathbb{R}$ relative to $\mathcal{A}_{p}$ with gap exceeding $\delta_{\mathcal{V}} \in (0, \delta_{\mathcal{V}}^*]$, where 
	$\delta_{\mathcal{V}}^* := ( \frac{2 r_o \|p_d-p_o\|}{\pi^2}    - \frac{\gamma_\theta}{2} )  \bar{\theta}_M^2$ with $\bar{\theta}_M := \max_{\bar{\theta}\in \Theta} |\bar{\theta}|$ and $0<\gamma_{\theta} < \frac{4r_o\|p_d-p_o\|}{\pi^2}$.
\end{lem}
\begin{proof}
	See Appendix \ref{sec:V_nav_synergistic}.
\end{proof}
Lemma \ref{lem:V_nav_synergistic} provides a choice for the design of parameters $(\Theta, \gamma_{\theta},\delta_{\mathcal{V}})$ for the synergistic navigation function $\mathcal{V}_{nav}$. Note that a decrease in the value of $\gamma_\theta$  results in an increase of the gap $\delta_{\mathcal{V}}$ (strengthening the robustness of the hybrid system). However, this may slow down the convergence rates of the overall closed-loop system as shown in \cite{wang2022hybrid}. Therefore, the scalar $\gamma_\theta$ needs to be carefully chosen via a tradeoff between the robustness and the convergence rates of the overall system.

\subsection{Hybrid Navigation Feedback Design}
Let $\mathcal{X}:= \mathcal{X}_p $ denote the state space and $x:= p \in \mathcal{X}$ denote the state. Similar to \eqref{eqn:new_affine_system}, system \eqref{eqn:system_robot} can be extended as follows:
% \begin{align} %\label{eqn:single_integrator_modified}
	% 	\begin{cases}
		% 			\dot{x} = u  \\
		% 			\dot{\theta} = \varpi(x,\theta)  
		% 	\end{cases}
	% \end{align}
\begin{align}\label{eqn:single_integrator_modified}
	%\left.
	%\begin{array}{ll}
	\begin{pmatrix}
		\dot{x} \\ 
		\dot{\theta}
	\end{pmatrix}
	& = \begin{pmatrix} 
		0 \\
		\varpi(x,\theta)
	\end{pmatrix} +
	\begin{pmatrix} 
		I_2 \\
		0
	\end{pmatrix}u
\end{align}
with $u=\kappa(x,\theta)$ and $\varpi(x,\theta)$ to be designed.  
Given a synergistic navigation function $\mathcal{V}_{nav}$ on $\mathcal{X}_p \times \mathbb{R}$   relative to $\mathcal{A}=\mathcal{A}_p=(p_d,0)$ as per Lemma \ref{lem:V_nav_synergistic}, we consider the following real-valued functions:
\begin{align}
	V(x,\theta)      & =   \mathcal{V}_{nav}(p,\theta)   \label{eqn:second_order_V}\\
	\kappa(x,\theta) & =  - k_p \nabla_{p} \mathcal{V}_{nav}(p,\theta)   \label{eqn:second_order_kappa} \\
	\varpi(x,\theta) & = - k_\theta \nabla_{\theta} \mathcal{V}_{nav}(p,\theta) \label{eqn:second_order_varpi}    
\end{align}
with constant scalars $k_p, k_\theta>0$. It is clear that $V$ is positive definite with respect to $\mathcal{A}$.
% \begin{align}
	% 	\mathcal{A} & := \{(x,\theta)\in  \mathcal{X} \times \mathbb{R}: (p,\theta) = \mathcal{A}_p, v = 0\}.  %\nonumber \\
	% 	        %& = \{(x,\theta)\in  \mathcal{X} \times \mathbb{R}: p = p_d, \theta=0, v = 0\}. 
	% 			\label{eqn:second_order_A}
	% \end{align}

\begin{prop}  \label{prop:synergistic_single}
	Consider the real-valued functions $V,\kappa,\varpi$ defined in \eqref{eqn:second_order_V}-\eqref{eqn:second_order_varpi}. Choose $(\Theta, \gamma_\theta, \delta_{\mathcal{V}})$ according to Lemma \ref{lem:V_nav_synergistic}, and the gains $k_\theta, k_p >0$.  Then, $(V,\kappa,\varpi, \Theta)$ is a synergistic feedback quadruple relative to $\mathcal{A}$ for system  \eqref{eqn:single_integrator_modified} with gap exceeding $\delta \in (0, \delta_{\mathcal{V}}]$.
\end{prop}

\begin{proof}
	See Appendix \ref{sec:synergistic_single}.
\end{proof}

Given the synergistic feedback quadruple $(V,\kappa,\varpi, \Theta)$ as per Proposition \ref{prop:synergistic_single}, we consider the following hybrid feedback %\eqref{eqn:hybrid_feedback1_design}
\begin{align}\label{eqn:second_order_hybrid_feedback}
	\underbrace{
		\begin{array}{ll}
			u            & = \kappa(x,\theta) \\
			\dot{\theta} & = \varpi(x,\theta)
		\end{array}
	}_{(x,\theta)\in \mathcal{F}}
	\underbrace{
		\begin{array}{ll}
			&\\
			\theta^+ & \in G_o(x,\theta)
		\end{array}
	}_{(x,\theta)\in \mathcal{J}}
\end{align}
where $G_o,\mathcal{F}, \mathcal{J}$ are designed as \eqref{eqn:hybrid_feedback1_design} with $V$ defined in \eqref{eqn:second_order_V}. Applying the results in Proposition \ref{prop:synergistic_single} and Theorem \ref{thm:thm1}, one concludes that the equilibrium point $\mathcal{A}$ is GAS for the hybrid closed-loop system \eqref{eqn:system_robot}. Moreover, safe navigation is guaranteed since the navigation function $\mathcal{V}_{nav}$ is bounded by the initial conditions for all $(t,j) \in \dom (x,\theta)$.

Substituting $\nabla_{p} \mathcal{V}_{nav}(p,\theta)$ in \eqref{eqn:def_grad_V'_p} into \eqref{eqn:second_order_kappa}, one can rewrite the hybrid feedback term $\kappa(x,\theta)$ in the form of \eqref{eqn:kappa_decomp} as:
\begin{align} \label{eqn:kappa_decomp2}
	\kappa(x,\theta) = \underbrace{-k_p \nabla_{p} V_{nav}(p)}_{\kappa_o(x)}  
	+  k_p \underbrace{(I_2-\mathcal{R}(\theta))\T (p_o-p_d)}_{\sigma(\theta)}   .
\end{align} 
The first term $\kappa_o(x)$ in \eqref{eqn:kappa_decomp2} is equivalent to the smooth time-invariant feedback, and the additional term $\sigma(\theta)$ in \eqref{eqn:kappa_decomp2} plays an important role in switching the feedback input $\kappa(x,\theta)$ when $x$ is close to one of the undesired equilibrium points of the closed-loop system. It is also noticed that the function $\sigma(\theta)$ is independent of the state $x$, which is a special case of the function $\sigma(x,\theta)$ in \eqref{eqn:kappa_decomp}.

\begin{lem}\label{lem:sigma} 
	Given $\Theta, \bar{\theta}_M$ according to Lemma \ref{lem:V_nav_synergistic},  the function $\sigma$ defined in \eqref{eqn:kappa_decomp2} satisfies $\max_{\bar{\theta}\in \Theta } \|\sigma({\theta}) - \sigma(\bar{\theta})\|^2  \leq  2c_k := 2(1-\cos(\bar{\theta}_M))  \|p_d-p_o\|^2$ for all $(x,\theta) \in \Psi_V \setminus \mathcal{A}$ with $\Psi_V = C_{\mathcal{V}_{nav}}$. 
\end{lem}
\begin{proof}
	See Appendix \ref{sec:sigma}.
\end{proof}
This lemma indicates that Assumption \ref{assum:sigma_bound} holds for the $\sigma(\theta)$ in \eqref{eqn:kappa_decomp2}. Thus, one can apply the smoothing mechanism developed in Section \ref{sec:hybrid_smooth} to improve the non-smooth hybrid feedback $\kappa(x,\theta)$ in \eqref{eqn:kappa_decomp2} by removing the discontinuities term $\sigma(\theta)$ in the control input $u$. 

Let us introduce the extended state $x_s := (p,\eta) \in \mathcal{X}_s := \mathcal{X}_p \times \mathbb{R}^2$. Similar to \eqref{eqn:new_affine_system}, system \eqref{eqn:system_robot} can be extended as follows: 
\begin{align}\label{eqn:single_integrator_modified2} 
	\begin{pmatrix}
		\dot{x} \\
		\dot{\eta} \\
		\dot{\theta}
	\end{pmatrix}
	& = \begin{pmatrix}
		-k_p \nabla_{p} V_{nav}(p) + k_p \eta  \\
		0 \\
		\varpi(x,\theta)
	\end{pmatrix} +
	\begin{pmatrix}
		0 \\
		I_2 \\
		0
	\end{pmatrix}u_s
\end{align}
with $u_s = \kappa_s(x_s,\theta)$ to be designed. Consider the following real-valued functions  
\begin{align}
	V_s(x_s,\theta) & =  V(x,\theta) + \frac{\gamma_s}{2}\|\eta - \sigma(\theta)\|^2 \label{eqn:def_V_s2} \\
	\kappa_s(x_s,\theta)  &= \mathcal{D}_t \sigma (\theta) -  k_\eta(\eta-\sigma(\theta)) - \frac{k_p}{\gamma_s} \nabla_p \mathcal{V}_{nav}(p,\theta)
	\label{eqn:def_kappa_s3}
\end{align}
with scalars $\gamma_s, k_\eta>0$, 
$$\mathcal{D}_t \sigma (\theta) =  - \mathcal{R}(\theta)\T \Delta (p_d-p_o) \varpi(x,\theta)$$ and $(V, \varpi, \sigma)$ defined in \eqref{eqn:second_order_V}, \eqref{eqn:second_order_varpi} and \eqref{eqn:kappa_decomp2}, respectively. It is not difficult to verify that $V_s$ is positive definite with respect to $\mathcal{A}_s :=  \{(x_s,\theta)\in \mathcal{X}_s \times \mathbb{R}: p=p_d, \theta=0, \eta = 0\}$ from the fact $\sigma(0) =0$. 
\begin{prop}  \label{prop:synergistic_single_smooth}
	Consider  the  real-valued functions $V_s,\kappa_s,\varpi$ defined in \eqref{eqn:def_V_s2}, \eqref{eqn:def_kappa_s3} and \eqref{eqn:second_order_varpi}, respectively. Choose $(\Theta, \gamma_\theta, \delta_{\mathcal{V}})$ according to Lemma \ref{lem:V_nav_synergistic}, $c_\kappa$ according to Lemma \ref{lem:sigma}, and $k_\theta, k_p, k_\eta >0, 0<\gamma_s < \frac{\delta_{\mathcal{V}}}{c_\kappa}$.  Then,  $(V_s,\kappa_s,\varpi, \Theta)$ is a synergistic feedback quadruple relative to   $\mathcal{A}_s$ for system \eqref{eqn:single_integrator_modified2} with gap exceeding $\delta_s\in (0, \delta_{\mathcal{V}}- \gamma_s c_\kappa]$.
\end{prop}
\begin{proof}
	See Appendix \ref{sec:synergistic_single_smooth}.
\end{proof}

Given the synergistic feedback quadruple $(V_s,\kappa_s,\varpi, \Theta)$ as per Proposition \ref{prop:synergistic_single_smooth}, we propose the following smooth hybrid feedback:
\begin{align}\label{eqn:hybrid_feedback_2}
	\underbrace{
		\begin{array}{ll}
			u            & = -k_p \nabla_{p} V_{nav}(p) + k_p \eta\\
			\dot{\eta}   & = \kappa_s(x_s,\theta) \\
			\dot{\theta} & = \varpi(x,\theta)
		\end{array}
	}_{(x_s,\theta)\in \mathcal{F}_s}
	\underbrace{
		\begin{array}{ll}
			&\\
			\eta^+   & = \eta\\
			\theta^+ & \in G_s(x_s,\theta)
		\end{array}
	}_{(x_s,\theta)\in \mathcal{J}_s}
\end{align}
where $(G_s, \mathcal{F}_s, \mathcal{J}_s)$ are designed as \eqref{eqn:hybrid_feedback2_design} with $V_s$ defined in \eqref{eqn:def_V_s2}. Note that the control input $u$ in \eqref{eqn:hybrid_feedback_2} is smooth and consists of the gradient of a nominal smooth navigation function $V_{nav}$ and an auxiliary variable $\eta$. 
Applying the results in Proposition \ref{prop:synergistic_single_smooth} and Theorem \ref{thm:thm2}, one concludes that the equilibrium point $\mathcal{A}_s$ is GAS for the hybrid closed-loop system \eqref{eqn:system_robot} and \eqref{eqn:hybrid_feedback_2}. Moreover, the collision-free navigation is guaranteed since the navigation function $\mathcal{V}_{nav}$ is bounded by the initial conditions for all $(t,j) \in \dom (x_s,\theta)$.

\subsection{Simulation}\label{sec:simulation} 
In this subsection, numerical simulations are presented to illustrate the performance of the proposed hybrid controllers by making use of the HyEQ Toolbox in MATLAB \cite{sanfelice2013toolbox}. 
For comparison purposes, three types of feedback controllers are considered: the hybrid controller \eqref{eqn:second_order_hybrid_feedback} referred to as 'Hybrid', the smooth hybrid controller \eqref{eqn:hybrid_feedback_2} referred to as 'Smooth-Hybrid', and the non-hybrid smooth feedback controller $u = \kappa_o(x)$ with $\kappa_o$ in \eqref{eqn:kappa_decomp2} referred to as 'Non-Hybrid'. The parameters for the barrier function $\phi$ in \eqref{eqn:Phi_def} are given as $\varepsilon = 0.1$ (m), $r_s = 0.5$ (m) and $\varrho = 15$.
We consider an obstacle centered at $p_{o}=(0,5)$ with radius  $r_{o}=2$ (m) in the workspace $\mathbb{R}^2$. The destination location is the origin and the initial conditions for $\theta$ and $\eta$ are set to zero.  All three feedback controllers share the same initial positions and feedback gain $k_p = 12$. The other parameters are given as: $\Theta=\{0.2\}, k_\theta = 0.02$, $\delta = 0.0365, \gamma_\theta = 2.0264, \gamma_s = 0.0659, \delta_s = 0.0036$ and $k_\eta=100$.

In the first simulation, the robot starts from the position $p(0)=(12,0)$, where convergence to the destination position is not guaranteed for the non-hybrid controller. In this case,
as one can notice from Fig. \ref{fig:simulation1}, only our nominal and smooth hybrid controllers guarantee asymptotic (collision-free) convergence to the destination position. The terminal position of the robot, with the non-hybrid controller, coincides with the undesired equilibrium point. As one can notice from Fig. \ref{fig:simulation2}, the variable $\theta$ converges to zero as stated in our main theorems.  Our second simulation in Fig. \ref{fig:simulation3} presents a comparison between three controllers from different initial positions, where both of our hybrid controllers guarantee asymptotic convergence.

\begin{figure}[!ht]
	\centering
	\subfloat[]{\includegraphics[width=0.8\linewidth]{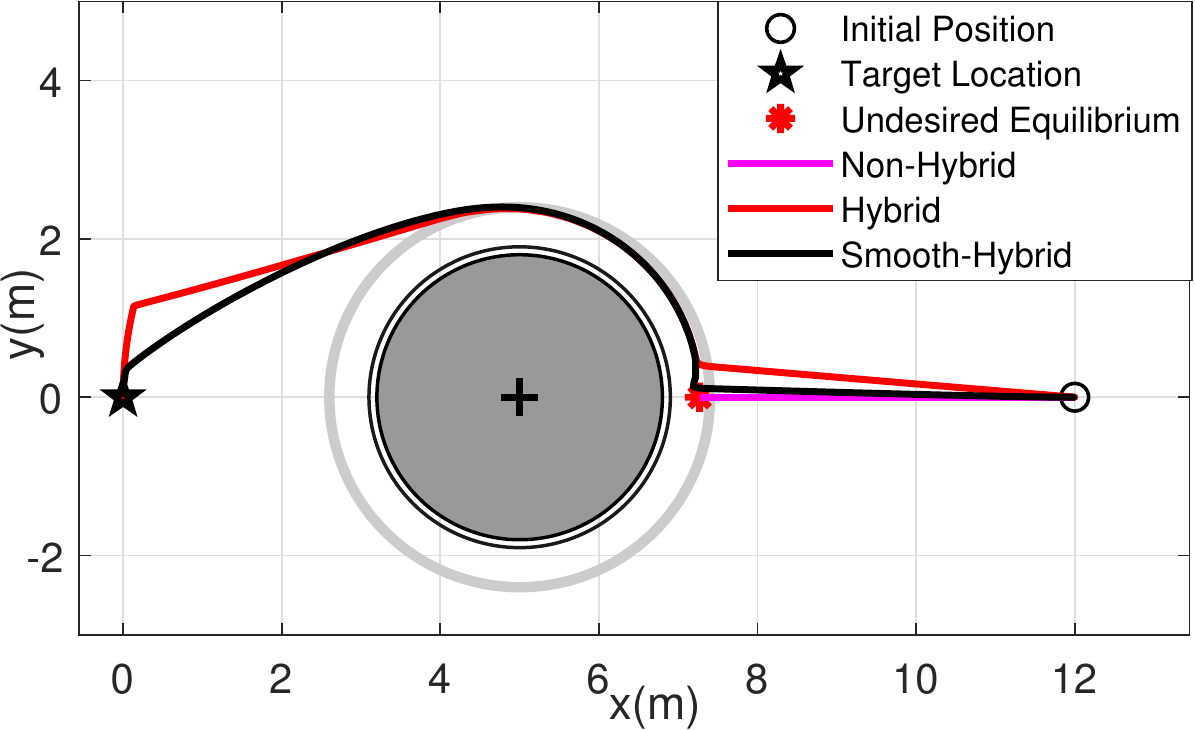}
		\label{fig:simulation1}}\\	
	\subfloat[]{\includegraphics[width=0.85\linewidth]{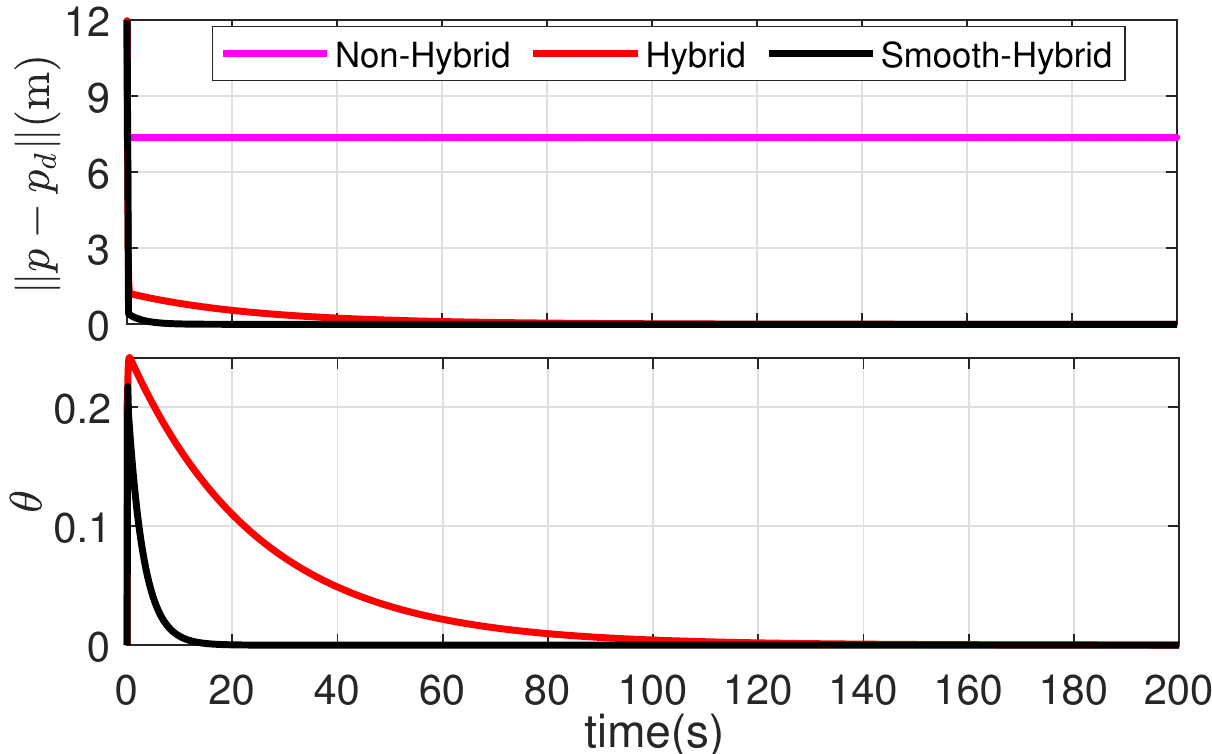}
		\label{fig:simulation2}}\\
	\subfloat[]{\includegraphics[width=0.8\linewidth]{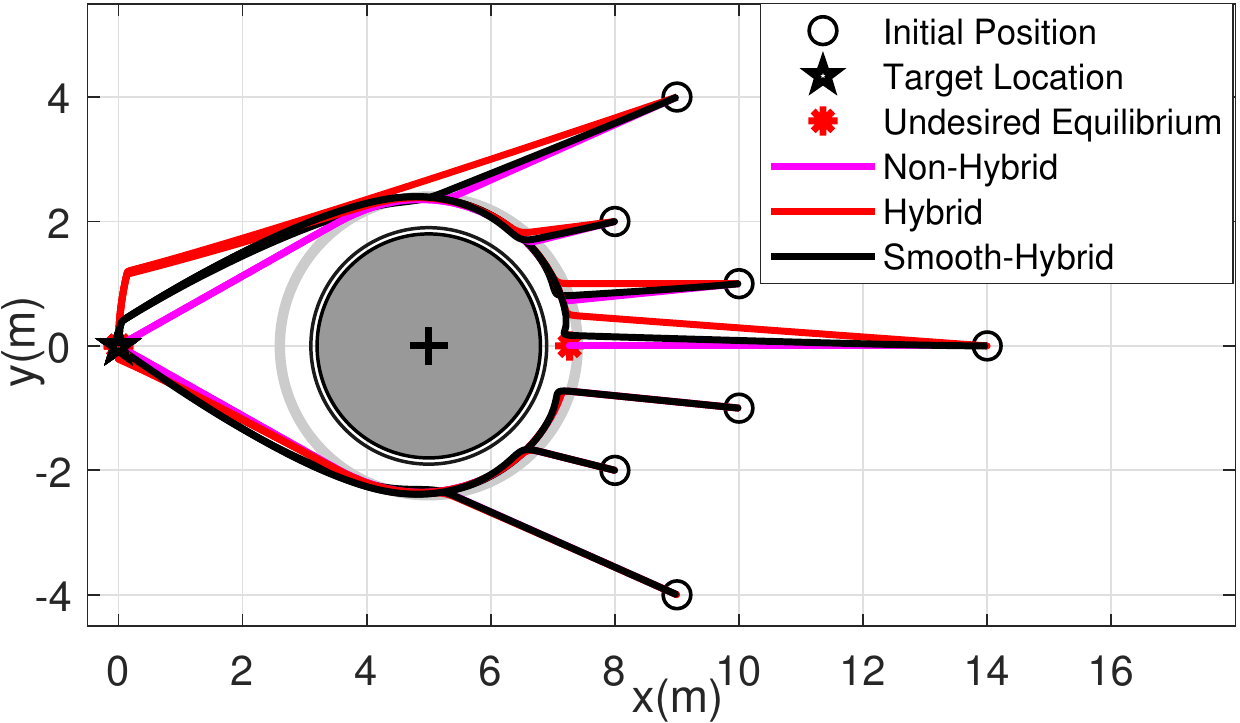}
		\label{fig:simulation3}} 
	\caption{(a) Trajectories of the robot using three different controllers with initial position $p(0)=(12,0)$; (b) Time evolution of the distance to the destination and the switching variable $\theta$; (c) Trajectories of the robot using three different controllers from different initial positions.} 
	\label{fig:simulation}
\end{figure}

\section{Conclusion}\label{sec:conclusion}
Over the past decade, synergistic hybrid feedback has been proven to be a powerful tool to achieve robust GAS for a class of affine nonlinear systems with topological constraints. The commonly used synergistic hybrid feedback approaches are designed based on a synergistic family of Lyapunov (or potential) functions and a logical index variable.  In this work, we generalize the concept of synergistic Lyapunov functions with a vector-valued switching variable, which is governed by hybrid dynamics. Unlike the logical index variable, this switching variable can either remain constant or change dynamically between jumps. Based on the generalized synergistic Lyapunov functions, we proposed a new synergistic hybrid feedback control scheme for a class of affine nonlinear systems leading to GAS. The main idea behind the proposed hybrid control strategy is to embed the state space in a higher dimensional space, which provides an easier handling of the equilibrium points of the overall closed-loop system through the hybrid dynamics of the switching variable leading to GAS guarantees. We also propose a smoothing mechanism for our synergistic hybrid feedback control strategy to get rid of the discontinuities in the hybrid feedback term. We also show that our hybrid approach can be extended to a larger class of systems through the integrator backstepping approach. Moreover, we have also introduced the concept of synergistic navigation functions, which we applied to the problem of safe navigation in planar environments, with global obstacle avoidance, for robots with single integrator dynamics.

%%%%%%%%%%%%%%%%%%%%%%%%%%%%%%%%%%%%%%%%%%%%%%%%%%%%%%%%%%%%%%%%%%%%%%%%%%%%%%%%
% \section*{Appendix}
\appendix

\subsection{Proof of Lemma \ref{lem:HBC}} \label{sec:HBC}
The flow set $\mathcal{F}$ and jump set $\mathcal{J}$ in \eqref{eqn:hybrid_feedback1_design} are closed subsets of $\mathcal{X}\times \mathbb{R}^r \subseteq \mathbb{R}^{n+r}$ since $\mu_{V,\Theta}$ is continuous (the composition of continuous functions), and satisfy $\mathcal{F} \cup \mathcal{J}=\mathcal{X}\times \mathbb{R}^r$. This shows the fulfillment of condition A1. Since  $\kappa$ and $\varpi$ are single-valued continuous functions on $\mathcal{F}$, the properties of outer semicontinuity, local boundedness, and convexity hold for the flow map $F$ in \eqref{eqn:hybrid_closed_F}, which shows the fulfillment of condition A2.  

Similar to the proof of \cite[Lemma 1]{casau2019hybrid}, let us consider a sequence $\{(x_i,\theta_i)\}_{i\in \mathbb{N}}\subset \mathcal{X}\times \mathbb{R}^r$ that converges to $(x,\theta)\in \mathcal{X}\times \mathbb{R}^r$, and let $\{\vartheta_i\}_{i\in \mathbb{N}}$ represent a sequence satisfying $V(x_i,\vartheta_i) = \min_{\bar{\theta}\in \Theta}V(x_i,\bar{\theta})$ (\ie, $\vartheta_i \in G_o(x_i,\theta_i)=\arg \min \{V(x,\bar{\theta}): \bar{\theta}\in \Theta\}$) for each $i\in \mathbb{N}$. Since $\Theta$ is compact (finite non-empty subset of $\mathbb{R}^r$) by assumption, there exists a subsequence $\{\vartheta_{i(k)}\}_{k\in \mathbb{N}}$ that converges to some $\vartheta\in \Theta$. Then, from \cite[Definition 5.9]{goebel2012hybrid} the set-valued function $G_o$ in \eqref{eqn:def_G_theta} is outer semicontinuous if $\vartheta \in G_o(x,\theta)$ at each $(x,\theta)\in \mathcal{J}$. Assume that there exists $\vartheta^*\in \mathbb{R}^r$ such that $V(x,\vartheta^*)<V(x,\theta)$, then, by continuity of $V$, there exists some $k^*\in \mathbb{N}$ such that $V(x_{i(k)},\vartheta^*)<V(x_{i(k)},\theta_i)$ for all $k>k^*$, which is a contradiction, since $\vartheta_i \in G_o(x_i, \theta_i)=\arg \min \{V(x,\bar{\theta}): \bar{\theta}\in \Theta\}$. Hence,  
one concludes that the set-valued function $G_o$ is outer semicontinuous, and from which the outer semicontinuity of the jump map $G$ in \eqref{eqn:hybrid_closed_F} relative to $\mathcal{J}$ follows. Moreover, the jump map $G$ is locally bounded relative to $\mathcal{J}$ since $G_o$ takes values over a finite discrete set $\Theta$ and the remaining component of $G$ is a single-valued continuous function on $\mathcal{J}$. This shows the fulfillment of condition A3, which completes the proof.

\subsection{Proof of Proposition \ref{prop:SLFS}} \label{sec:SLFS}
From the definitions of the function $V_s$ in \eqref{eqn:def_V_s} and the set $\mathcal{A}_s$ in \eqref{eqn:def_A_s}, applying the properties of the synergistic feedback quadruple $(V,\kappa,\varpi, \Theta)$ in Definition \ref{defn:synergistic_feedback}, one can verify that $V_s$ is positive definite with respect to $\mathcal{A}_s$, and for each $\epsilon\geq 0$ the sub-level set  $\mho_{V_s}(\epsilon):= \{(x_s,\theta)\in \mathcal{X}_s \times \mathbb{R}^r: V_s(x_s,\theta)\leq \epsilon\} $ is compact. Hence, the conditions C1) and C2) in Definition \ref{defn:synergistic_feedback} for the quadruple $(V_s,\kappa_s,\varpi, \Theta)$ are satisfied.

Next, we are going to verify the condition  C3) in Definition \ref{defn:synergistic_feedback} for the quadruple $(V_s,\kappa_s,\varpi, \Theta)$. Letting 
$$\tilde{\eta}:=\eta-\sigma(x,\theta)$$
for the sake  of simplicity, the function $V_s$ defined in \eqref{eqn:def_V_s} can be rewritten as
\begin{align*}
	V_s(x_s,\theta) = V(x,\theta) + \frac{\gamma_s}{2} \tilde{\eta}\T \tilde{\eta}.
\end{align*}
From \eqref{eqn:kappa_decomp}, \eqref{eqn:smooth_affine_system} and \eqref{eqn:def_kappa_s}, for all $(x_s,\theta)\in \mathcal{X}_s \times \mathbb{R}^r$,  one obtains
\begin{align}
	& \langle \nabla V_s(x_s,\theta), f_s(x_s,\theta) + g_s(x_s,\theta)\kappa_s(x_s,\theta)\rangle  \nonumber \\
	& = \langle \nabla_x V(x,\theta), f(x) + g(x)(\varsigma(x)+\Upsilon(x)\eta)\rangle  \nonumber \\
	& \quad + \langle \nabla_\theta V(x,\theta), \varpi(x,\theta)\rangle   +  \gamma_s\tilde{\eta}\T \left( \kappa_s(x_s,\theta) - \mathcal{D}_t \sigma(x,\theta)\right)    \nonumber \\
	& = \big\langle \nabla_x V(x,\theta), f(x) + g(x)\kappa(x,\theta) + g(x) \Upsilon(x)\tilde{\eta} \big\rangle    \nonumber\\
	& \quad + \langle \nabla_\theta V(x,\theta), \varpi(x,\theta)\rangle  \nonumber\\
	& \quad +  \tilde{\eta}\T ( \gamma_s k_\eta\tilde{\eta}   -  \Upsilon(x)\T g(x)\T \nabla_x V(x,\theta))   \nonumber \\
	& = \langle \nabla_x V(x,\theta), f(x) + g(x)\kappa(x,\theta) \rangle +  \langle \nabla_\theta V(x,\theta), \varpi(x,\theta)\rangle \nonumber \\
	& \quad   - \gamma_s k_\eta   \tilde{\eta}\T \tilde{\eta} \nonumber\\
	& = \langle \nabla V(x,\theta), f_c(x,\theta) + g_c(x,\theta)\kappa(x,\theta) \rangle  - \gamma_s k_\eta   \tilde{\eta}\T \tilde{\eta} \nonumber \\
	& \leq 0 \label{eqn:prop_kappa_s}
\end{align}
where we made use of the facts $\varsigma(x)+\Upsilon(x)\eta  = \kappa(x,\theta) + \Upsilon(x)\tilde{\eta}$ and $\langle \nabla V(x,\theta), f_c(x,\theta) + g_c(x,\theta)\kappa(x,\theta) \rangle =\langle \nabla_x V(x,\theta), f(x) + g(x)\kappa(x,\theta) \rangle +  \langle \nabla_\theta V(x,\theta), \varpi(x,\theta)\rangle \leq  0 $ from \eqref{eqn:prop_kappa} for all $(x,\theta)\in \mathcal{X} \times \mathbb{R}^r$.
Then, one concludes that the condition C3) in Definition \ref{defn:synergistic_feedback} is satisfied for the quadruple $(V_s,\kappa_s,\varpi, \Theta)$.

Now, we need to verify the condition  C4) in Definition \ref{defn:synergistic_feedback} for the quadruple $(V_s,\kappa_s,\varpi, \Theta)$. Similar to the definition of $\mathcal{E}$, we define
\begin{multline}
	\mathcal{E}_s  :=  \{(x_s,\theta)\in \mathcal{X}_s \times \mathbb{R}^r:
	\langle \nabla V_s(x_s,\theta), \\ f_s(x_s,\theta) + g_s(x_s,\theta)\kappa_s(x_s,\theta)\rangle = 0 \}.
	\label{eqn:def_E_s}
\end{multline}
From the definition of $\mathcal{E}$ in \eqref{eqn:def_E} and \eqref{eqn:prop_kappa_s}, one can rewrite $\mathcal{E}_s$ as
\begin{align*}
	\mathcal{E}_s
	& = \{(x_s,\theta)\in \mathcal{X}_s \times \mathbb{R}^r: \eta-\sigma(x,\theta) = 0, \nonumber \\
	& \qquad \langle \nabla V(x,\theta), f_c(x,\theta) + g_c(x,\theta)\kappa(x,\theta) \rangle = 0 \} \nonumber\\
	& = \{(x_s,\theta)\in \mathcal{X}_s \times \mathbb{R}^r: (x,\theta)\in \mathcal{E}, \eta = \sigma(x,\theta)\} . %\label{eqn:def_E_s2}
\end{align*}
Let $\Psi_{V_s} \subseteq \mathcal{E}_s$ denote the  largest weakly invariant set for system \eqref{eqn:smooth_affine_system} with $(x_s,\theta) \in  \mathcal{E}_s$.
From the definitions of $f_s, g_s$ in \eqref{eqn:smooth_affine_system} and $\kappa_s$ in \eqref{eqn:def_kappa_s}, the time-derivative of $x_s$ can be explicitly expressed as
\begin{align*}
	\dot{x}_s = \begin{pmatrix}
		f(x) + g(x)(\kappa(x,\theta) + \Upsilon(x)\tilde{\eta}) \\
		\mathcal{D}_t{\sigma}(x,\theta) -  k_\eta\tilde{\eta}-  \frac{1}{\gamma_s} \Upsilon(x)\T g(x)\T \nabla_x V(x,\theta)
	\end{pmatrix}
\end{align*}
with $\tilde{\eta}=\eta - \sigma(x,\theta)$ and $\mathcal{D}_t{\sigma}(x,\theta) = \frac{d}{dt}{\sigma}(x,\theta)$ by definition.
From $\tilde{\eta}\equiv 0$, \ie, $\eta \equiv \sigma(x,\theta)$, it follows that $   \frac{d}{dt}(\eta - {\sigma}(x,\theta))\equiv 0$ and $\Upsilon(x)\T g(x)\T \nabla_x V(x,\theta) \equiv 0$. Therefore, the largest weakly invariant set for system   \eqref{eqn:smooth_affine_system} with $(x_s,\theta) \in  \mathcal{E}_s$ satisfies $
\Psi_{V_s}   \subseteq  \{(x_s,\theta)\in \mathcal{X}_s \times \mathbb{R}^r: (x,\theta)\in \mathcal{E} \cap \mathcal{W}, \eta = \sigma(x,\theta)\}
$ with $ 
\mathcal{W}: = \{(x,\theta) \in \mathcal{X}\times \mathbb{R}^r: \Upsilon(x)\T g(x)\T \nabla_{x} V(x,\theta) = 0\}. %\label{eqn:def_W}
$
Let $\bar{\Psi}_V \subseteq \mathcal{E} \cap \mathcal{W}$ denote  the largest weakly invariant set for system \eqref{eqn:new_affine_system} with $(x,\theta) \in  \mathcal{E}\cap \mathcal{W}$. It follows that $\bar{\Psi}_V \subseteq \Psi_V$, and then one can further show that
$
\Psi_{V_s} =  \{(x_s,\theta)\in \mathcal{X}_s \times \mathbb{R}^r: (x,\theta)\in \bar{\Psi}_V, \eta = \sigma(x,\theta)\}.
$
Thus, for each $(x_s,\theta)\in \Psi_{V_s} $ one has  $(x,\theta)\in \bar{\Psi}_V \setminus \mathcal{A}$ and $\tilde{\eta} = \eta - \sigma(x,\theta)=0$. Consequently, from  \eqref{eqn:mu_V_delta}  and the condition in Assumption \ref{assum:sigma_bound} one obtains
\begin{align}
	& V_s(x_s,\theta) - \min_{\bar{\theta}\in \Theta} V_s(x_s,\bar{\theta})  \nonumber \\
	& = V(x,\theta) + \frac{\gamma_s}{2}\|\tilde{\eta}\|^2  - \min_{\bar{\theta}\in \Theta} \left( V(x,\bar{\theta}) + \frac{\gamma_s}{2}\|\eta - \sigma(x,\bar{\theta})\|^2\right) \nonumber\\
	& = V(x,\theta)  - \min_{\bar{\theta}\in \Theta} \left( V(x,\bar{\theta}) + \frac{\gamma_s}{2}\|\sigma(x,\theta) - \sigma(x,\bar{\theta})\|^2\right) \nonumber\\
	& \geq V(x,\theta) - \min_{\bar{\theta}\in \Theta}   V(x,\bar{\theta})  - \frac{\gamma_s}{2}\max_{\bar{\theta}\in\Theta}\|\sigma(x,\theta) - \sigma(x,\bar{\theta})\|^2  \nonumber \\
	& \geq \mu_{V,\Theta}(x,\theta)  -  \gamma_s c_k  \label{eqn:V-minV}
\end{align}
where we made use of the inequality $
\min_{\bar{\theta}\in \Theta} ( V(x,\bar{\theta}) + \frac{\gamma_s}{2}\|\eta - \sigma(x,\bar{\theta})\|^2)      \leq \min_{\bar{\theta}\in \Theta}  V(x,\bar{\theta}) + \frac{\gamma_s}{2}\max_{\bar{\theta}\in\Theta}\|\eta - \sigma(x,\bar{\theta})\|^2$. 
Similar to the definition of $\delta_{V,\Theta}$ in \eqref{eqn:mu_V_delta}, it follows from  \eqref{eqn:V-minV}  that
\begin{align}
	\delta_{V_s,\Theta}(x_s,\theta) & =  \inf_{(x_s,\theta)\in \Psi_{V_s} \setminus \mathcal{A}_s} ( V_s(x_s,\theta) - \min_{\bar{\theta}\in \Theta} V_s(x_s,\bar{\theta}))   \nonumber \\
	& \geq \inf_{(x,\theta)\in \bar{\Psi}_V \setminus \mathcal{A}} \mu_{V,\Theta}(x,\theta)   -  \gamma_s c_k    \nonumber\\
	& \geq \inf_{(x,\theta)\in  \Psi_V \setminus \mathcal{A}} \mu_{V,\Theta}(x,\theta)   -  \gamma_s c_k    \nonumber\\
	& =  \delta_{V,\Theta}(x,\theta) -  \gamma_s c_\kappa\nonumber\\
	& > \delta -\gamma_s c_\kappa \geq \delta_s \label{eqn:delta_V_s}
\end{align}
where we made use of the facts $\gamma_s < {\delta}/{c_k}$, $\delta_s\in (0,\delta- \gamma_s c_k ]$, and $\inf_{(x,\theta)\in \bar{\Psi}_V \setminus \mathcal{A}} \mu_{V,\Theta}(x,\theta) \geq \inf_{(x,\theta)\in  \Psi_V \setminus \mathcal{A}}\mu_{V,\Theta}(x,\theta) $ since $ \bar{\Psi}_V \subseteq \Psi_V $. 
Hence, for the quadruple $(V_s,\kappa_s,\varpi, \Theta)$ the condition C4) in Definition \ref{defn:synergistic_feedback} follows immediately, which completes the proof. 

\subsection{Proof of Proposition \ref{prop:SLFS_backstepping}} \label{sec:SLFS_backstepping}
The proof of Proposition \ref{prop:SLFS_backstepping} is similar to the proof of   Proposition \ref{prop:SLFS}. One can also show that the conditions C1) and C2) in Definition \ref{defn:synergistic_feedback} for the quadruple $(V_b,\kappa_b,\varpi, \Theta)$ are satisfied. For all $(x_b,\theta) \in \mathcal{X}_b \times \mathbb{R}^r$, one can show that
\begin{align}
	& \langle \nabla V_b(x_b,\theta), f_b(x_b,\theta) + g_b(x_b,\theta)\kappa_b(x_b,\theta)\rangle  \nonumber\\
	& = \langle \nabla_x V(x,\theta), f(x) + g(x)u \rangle   + \langle \nabla_\theta V(x,\theta), \varpi(x,\theta) \rangle  \nonumber \\
	& \quad  + \gamma_s \tilde{\eta}\T \left(\kappa_s(x_s,\theta) - \mathcal{D}_t \sigma(x,\theta) \right)  \nonumber\\
	& \quad  + \gamma_b (u-\bar{\kappa}(x,\eta))\T \left(\kappa_b(x_b,\theta) - \mathcal{D}_t \bar{\kappa}(x,\eta) \right)  \nonumber\\
	& = \langle \nabla_x V(x,\theta), f(x) + g(x)\left( \kappa(x,\theta)   + \Upsilon(x)\tilde{\eta}   \right)  \rangle  \nonumber\\
	& \quad  + \langle \nabla_\theta V(x,\theta), \varpi(x,\theta) \rangle    + \gamma_s \tilde{\eta}\T \left(\kappa_s(x_s,\theta) - \mathcal{D}_t \sigma(x,\theta) \right)  \nonumber \\
	& \quad +  \langle \nabla_x V(x,\theta),   g(x) (u- \bar{\kappa}(x,\eta) )  \rangle  -  (u-\bar{\kappa}(x,\eta))\T \nonumber \\
	& \qquad  (   \gamma_bk_b(u-\bar{\kappa}(x,\eta)) + g(x)\T \nabla_x V(x,\theta) )  \nonumber\\
	& = \langle \nabla V(x,\theta), f_c(x,\theta) + g_c(x,\theta)\kappa(x,\theta) \rangle  \nonumber\\
	& \quad   - \gamma_s k_\eta \| \tilde{\eta}\|^2 - \gamma_bk_b\|u-\bar{\kappa}(x,\eta)\|^2 \leq 0  \label{eqn:prop_kappa_b}
\end{align}
where we made use of the  facts $\tilde{\eta} = \eta - \sigma(x,\theta)$,
$
\bar{\kappa}(x,\eta) %&=   \varsigma(x) + \Upsilon(x)\eta\\
=    \varsigma(x) + \Upsilon(x)(\eta - \sigma(x,\theta) + \sigma(x,\theta) )
=   \kappa(x,\theta) +   \Upsilon(x)\tilde{\eta}
$,
and $\langle \nabla V(x,\theta), f_c(x,\theta) + g_c(x,\theta)\kappa(x,\theta) \rangle \leq  0$ for all $(x,\theta)\in \mathcal{X} \times \mathbb{R}^r$.
Therefore, one has condition C3) in Definition \ref{defn:synergistic_feedback} for the quadruple $(V_b,\kappa_b,\varpi, \Theta)$.
Similar to the definitions of  $\mathcal{E}$ in \eqref{eqn:def_E} and $\mathcal{E}_s$ in \eqref{eqn:def_E_s}, we define
\begin{multline}
	\mathcal{E}_b  :=  \{(x_b,\theta)\in \mathcal{X}_b \times \mathbb{R}^r:
	\langle \nabla V_b(x_b,\theta), \\ f_b(x_b,\theta) + g_b(x_b,\theta)\kappa_b(x_b,\theta)\rangle = 0 \}.  \label{eqn:def_E_b}
\end{multline}
From the definition of $\mathcal{E}$ in \eqref{eqn:def_E} and \eqref{eqn:prop_kappa_b}, one can rewrite $\mathcal{E}_b$ as
\begin{align}
	\mathcal{E}_b
	& = \{(x_b,\theta)\in \mathcal{X}_b \times \mathbb{R}^r: \eta=\sigma(x,\theta), u = \bar{\kappa}(x,\eta), \nonumber\\
	& \qquad \langle \nabla V(x,\theta), f_c(x,\theta) + g_c(x,\theta)\kappa(x,\theta) \rangle = 0 \} \nonumber\\
	& = \{(x_b,\theta)\in \mathcal{X}_b \times \mathbb{R}^r: (x_s,\theta) \in \mathcal{A}_s,  u = \bar{\kappa}(x,\eta)\}.
\end{align}
Let $\Psi_{V_b} \subseteq \mathcal{E}_b$ denote the  largest weakly invariant set for the   system \eqref{eqn:backstepping_affine_system} with $(x_b,\theta) \in  \mathcal{E}_b$.
From the definitions of $f_b, g_b$ in \eqref{eqn:backstepping_affine_system} and $\kappa_b$ in \eqref{eqn:def_kappa_b}, the time-derivative of $x_b$ can be explicitly expressed as
\begin{align*}
	\dot{x}_b  =  \begin{pmatrix}
		f(x) + g(x)(   \kappa(x,\theta) +   \Upsilon(x)\tilde{\eta} + u- \bar{\kappa}(x,\eta) )\\
		\mathcal{D}_t{\sigma}(x,\theta) -  k_\eta\tilde{\eta}-  \frac{1}{\gamma_s} \Upsilon(x)\T g(x)\T \nabla_x V(x,\theta) \\
		\mathcal{D}_t \bar{\kappa}(x,\eta) -  k_b(u-\bar{\kappa}(x,\eta))  -  \frac{1}{\gamma_b} g(x)\T \nabla_x V(x,\theta) \\
		\varpi(x,\theta)
	\end{pmatrix}
\end{align*}
where $\tilde{\eta}=\eta - \sigma(x,\theta)$, $\mathcal{D}_t{\sigma}(x,\theta) = \frac{d}{dt}{\sigma}(x,\theta)$ and $\mathcal{D}_t \bar{\kappa}(x,\eta) = \frac{d}{dt}\bar{\kappa} (x,\eta)$ (between the jumps of $\theta$).
From $\tilde{\eta}\equiv 0$ and $u \equiv \bar{\kappa}(x,\eta)$, it follows that $\frac{d}{dt}(\eta - {\sigma}(x,\theta))\equiv 0$ and $\frac{d}{dt} (u - \bar{\kappa}(x,\eta)) \equiv 0$, and then one obtains $g(x)\T \nabla_x V(x,\theta) \equiv 0$. Therefore, the largest weakly invariant set for system   \eqref{eqn:backstepping_affine_system} with $(x_b,\theta) \in  \mathcal{E}_b$ satisfies $ 
\Psi_{V_b}  \subseteq  \{(x_b,\theta)\in \mathcal{X}_b \times \mathbb{R}^r: (x,\theta)\in \mathcal{E} \cap \mathcal{W}',  \eta = \sigma(x,\theta), u = \bar{\kappa}(x,\eta)\}$ 
with $
\mathcal{W}': = \{(x,\theta) \in \mathcal{X}\times \mathbb{R}^r:  g(x)\T \nabla_{x} V(x,\theta) = 0\}. %\label{eqn:def_W}
$
Let $\bar{\Psi}_V' \subset \mathcal{E} \cap \mathcal{W}'$ denote  the largest weakly invariant set for system \eqref{eqn:new_affine_system} with $(x,\theta) \in  \mathcal{E}\cap \mathcal{W}'$. It follows that $\bar{\Psi}_V' \subseteq \bar{\Psi}_V \subseteq \Psi_V$  and  one can further show that
$       
\Psi_{V_b} =  \{(x_b,\theta)\in \mathcal{X}_b \times \mathbb{R}^r: (x,\theta)\in \bar{\Psi}_V',
\eta = \sigma(x,\theta), u = \bar{\kappa}(x,\eta)\}. $ 
Then, for each $(x_b,\theta)\in \Psi_{V_b} $ one has  $(x,\theta)\in \bar{\Psi}_V' \setminus \mathcal{A}$ and $\eta = \sigma(x,\theta)$ and $u = \bar{\kappa}(x,\eta)$. Consequently, from  \eqref{eqn:def_V_b} and \eqref{eqn:V-minV} one obtains
\begin{align}
	V_b(x_b,\theta) - \min_{\bar{\theta}\in \Theta} V_b(x_b,\bar{\theta})
	& = V_s(x_s,\theta) - \min_{\bar{\theta}\in \Theta} V_s(x_s,\bar{\theta})  \nonumber \\
	& \geq \mu_{V,\Theta}(x,\theta)  -  \gamma_s c_k.  \label{eqn:V-minV2}
\end{align}
Then, similar to the definition of $\delta_{V,\Theta}$ in \eqref{eqn:mu_V_delta}, it follows from \eqref{eqn:delta_V_s} and \eqref{eqn:V-minV2} that
\begin{align}
	\delta_{V_b,\Theta}(x_b,\theta) & =  \inf_{(x_b,\theta)\in \Psi_{V_b} \setminus \mathcal{A}_b} ( V_b(x_b,\theta) - \min_{\bar{\theta}\in \Theta} V_b(x_b,\bar{\theta}))   \nonumber \\
	%& \geq \inf_{(x,\theta)\in \bar{\Psi}_V \setminus \mathcal{A}} \mu_{V,\Theta}(x,\theta)   -  \gamma_s c_k    \nonumber \\ 
	& \geq \inf_{(x,\theta)\in  \Psi_V \setminus \mathcal{A}} \mu_{V,\Theta}(x,\theta)   -  \gamma_s c_k    \nonumber \\
	%& =  \delta_{V,\Theta}(x,\theta) -  \gamma_s c_\kappa\nonumber\\
	& > \delta -\gamma_s c_\kappa \geq \delta_b \label{eqn:delta_V_b}.
\end{align} 
Therefore, the condition C4) in Definition \ref{defn:synergistic_feedback} is satisfied for the quadruple $(V_b,\kappa_b,\varpi, \Theta)$, which completes the proof.

\subsection{Proof of Lemma \ref{lem:grad_V_nav2}} \label{sec:grad_V_nav2}
For the sake of simplicity, let $\bar{p} = \mathcal{T}(p,\theta)$ with $\mathcal{T}$ defined in \eqref{eqn:def_mathcal_T} and $z = d_o(\bar{p}) =d_o(p) = \|p-p_o\|-r_o$. Then, the real-valued function  $\mathcal{V}_{nav}$ in \eqref{eqn:def_synergistic_navigation} can be simplified as
\begin{align}
	\mathcal{V}_{nav}(p,\theta) & = V_{nav}(\bar{p}) + \frac{\gamma_\theta}{2} \theta^2.
	\label{eqn:def_synergistic_navigation3}
\end{align}	
Applying the definition of gradient, the time-derivative of $\mathcal{V}_{nav}$ along the trajectories of $\dot{p}=v$ and $\dot{\theta} = w$ is given as
\begin{align}
	\dot{\mathcal{V}}_{nav}(p,\theta)  
	& = \langle \nabla_{p} \mathcal{V}_{nav}(p,\theta), v \rangle  + \langle \nabla_{\theta} \mathcal{V}_{nav}(p,\theta), w \rangle.
	\label{eqn:dot_V_nav1}
\end{align}
On the other hand, from \eqref{eqn:def_synergistic_navigation3}, the time-derivative of $\mathcal{V}_{nav}$ can be explicitly written as  
\begin{align}
	&\dot{\mathcal{V}}_{nav}(p,\theta) \nonumber \\
	& \qquad = \langle \nabla_{\bar{p}} V_{nav}(\bar{p}),  \mathcal{D}_p\mathcal{T}(p,\theta) v + \mathcal{D}_\theta \mathcal{T}(p,\theta) w \rangle +  \langle \gamma_\theta  \theta, w \rangle \nonumber \\
	& \qquad = \langle \mathcal{D}_p\mathcal{T}(p,\theta)\T \nabla_{\bar{p}} V_{nav}(\bar{p}),   v \rangle  \nonumber \\
	& \qquad  \qquad + \langle \gamma_\theta  \theta + \mathcal{D}_\theta \mathcal{T}(p,\theta)\T \nabla_{\bar{p}} V_{nav}(\bar{p})), w \rangle
	\label{eqn:dot_V_nav2}
\end{align}	
where we have made use of the facts $\dot{\bar{p}} = \dot{\mathcal{T}}(p,\theta)=\mathcal{D}_p \mathcal{T}(p,\theta) v + \mathcal{D}_\theta \mathcal{T}(p,\theta) w $ with $\mathcal{D}_p\mathcal{T}(p,\theta)$ and $\mathcal{D}_\theta \mathcal{T}(p,\theta)$ denoting the Jacobian matrices of $\mathcal{T}$ with respect to $p$ and $\theta$, respectively. Hence, from \eqref{eqn:dot_V_nav1} and \eqref{eqn:dot_V_nav2}, one obtains the gradients $\nabla_{p} \mathcal{V}_{nav}(p,\theta)$ and $\nabla_{\theta} \mathcal{V}_{nav}(p,\theta)$ as follows: 
\begin{subequations}\label{eqn:def_grad_V}
	\begin{align}
		\nabla_{p} \mathcal{V}_{nav}(p,\theta)      & = \mathcal{D}_p\mathcal{T}(p,\theta)\T \nabla_{\bar{p}} V_{nav}(\bar{p}) \label{eqn:def_grad_p_V} \\
		\nabla_{\theta} \mathcal{V}_{nav}(p,\theta) & = \gamma_\theta \theta +  \mathcal{D}_\theta \mathcal{T}(p,\theta)\T \nabla_{\bar{p}} V_{nav}(\bar{p}) \label{eqn:def_grad_zeta_V}.
	\end{align}
\end{subequations}
From \eqref{eqn:gradient_V_nav}, one can explicitly write $\nabla_{\bar{p}} V_{nav}(\bar{p})$ as
\begin{align}
	\nabla_{\bar{p}} V_{nav}(\bar{p}) %&= \bar{p} - p_d + \varrho \nabla_{d(\bar{p})} \phi(d(\bar{p})) \frac{\bar{p}-p_o}{\|\bar{p}-p_o\|} 
	& = \bar{p} - p_d + \varrho \nabla_{z} \phi(z) \frac{\bar{p}-p_o}{\|\bar{p}-p_o\|}.
	\label{eqn:def_grad_bar_p_V}
\end{align}
Moreover, applying the definition of $\mathcal{T}$ in \eqref{eqn:def_mathcal_T}, the time-derivative of $\mathcal{T}$ along the trajectories of $\dot{p}=v$ and $\dot{\theta} = w$ can be written as 
\begin{align*} %\mathcal{T}(p,\theta)=   p_o + \mathcal{R}(\theta) (p-p_o)
	\frac{d}{dt}{\mathcal{T}}(p,\theta) %&= \dot{\mathcal{R}}(\theta) (p-p_o) + \mathcal{R}(\theta) \dot{p} \nonumber \\
	& = \mathcal{R}(\theta) \Delta (p-p_o)  w + \mathcal{R}(\theta) v %\nonumber \\
	%& = \mathcal{R}(\theta) (w \Delta (p-p_o) + v)
	%\label{eqn:def_dot_T_2}
\end{align*}
where $\Delta$ is defined in \eqref{eqn:def_R}, and we made use of the facts  $\frac{d}{dt}{\mathcal{R}}(\theta) = \frac{d}{dt} \exp(\theta \Delta) = w \mathcal{R}(\theta) \Delta$. Consequently, the Jacobian matrices of $\mathcal{T}$ in \eqref{eqn:def_mathcal_T} are given as 
\begin{align}\label{eqn:def_D_T}
	\mathcal{D}_p\mathcal{T}(p,\theta)      = \mathcal{R}(\theta), \quad 
	\mathcal{D}_\theta \mathcal{T}(p,\theta)   = \mathcal{R}(\theta) \Delta (p-p_o) .  
\end{align} 	  
Substituting  \eqref{eqn:def_grad_bar_p_V} and \eqref{eqn:def_D_T} into \eqref{eqn:def_grad_V}, one obtains
\begin{align}
	\nabla_{p} \mathcal{V}_{nav}(p,\theta)  
	& = \mathcal{R}(\theta) \T (\bar{p} - p_d )  + \varrho \nabla_{z} \phi(z) \frac{p-p_o}{\|p-p_o\|} \nonumber\\
	& = p-p_d + \varrho \nabla_{z} \phi(z) \frac{p-p_o}{\|p-p_o\|}   \nonumber \\
	& \qquad \qquad - (I_2-\mathcal{R}(\theta))\T (p_o-p_d) 
	\label{eqn:def_grad_p_V_2}
	\\
	\nabla_{\theta} \mathcal{V}_{nav}(p,\theta) 
	& = \gamma_\theta \theta - (p-p_o)\T \Delta \mathcal{R}(\theta)\T (\bar{p} - p_d)               \nonumber\\
	& = \gamma_\theta \theta - (p-p_o)\T \Delta \mathcal{R}(\theta)\T (p_o - p_d)
	\label{eqn:def_grad_zeta_V_2}
\end{align}
where we made use of the facts $\bar{p} - p_o = \mathcal{R}(\theta) (p - p_o)$, % $\bar{p} = \mathcal{T}(p,\theta) = p_o + \mathcal{R}(\theta)(p-p_o)$ and
$
\mathcal{R}(\theta) \T (\bar{p} - p_d )
=  p  -p_o + \mathcal{R}(\theta) \T (p_o-p_d)
=  p-p_d - (I_2-\mathcal{R}(\theta))\T (p_o-p_d)
$, $(p-p_o)\T\Delta (p-p_o)=0$ and $\Delta\T = -\Delta$. % from the skew-symmetric property of $\Delta$. 
From \eqref{eqn:def_grad_p_V_2} and \eqref{eqn:def_grad_zeta_V_2} as well as the definition of $\nabla_p V_{nav}(p)$ in \eqref{eqn:gradient_V_nav},
one concludes \eqref{eqn:def_grad_V'_p} and \eqref{eqn:def_grad_V'_zeta}.

By definition, the set of the critical points of $\mathcal{V}_{nav}$ is defined as
\begin{multline}
	C_{\mathcal{V}_{nav}}:=  \{(p,\theta) \in \mathcal{X}_p \times \mathbb{R}: \\
	\nabla_{p} \mathcal{V}_{nav}(p,\theta) =0, \nabla_{\theta} \mathcal{V}_{nav}(p,\theta)=0\} \label{eqn:C_V_nav}.
\end{multline} 
From \eqref{eqn:def_grad_p_V} and \eqref{eqn:def_D_T},  $\nabla_{p} \mathcal{V}_{nav}(p,\theta) =0$ implies that $\nabla_{\bar{p}} V_{nav}(\bar{p})=0$ since $\mathcal{D}_p\mathcal{T}(p,\theta)  = \mathcal{R}(\theta) \in SO(2)$. Substituting $\nabla_{\bar{p}} V_{nav}(\bar{p})=0$ in \eqref{eqn:def_grad_zeta_V},  $\nabla_{\theta} \mathcal{V}_{nav}(p,\theta) =0$ implies that $\theta=0$ since $\gamma_\theta>0$. Moreover, from $\theta = 0$ and the definition of $\mathcal{T}$ in \eqref{eqn:def_mathcal_T}, one has $\bar{p}=\mathcal{T}(p,\theta=0) = p$, which further implies that $\nabla_{\bar{p}} V_{nav}(\bar{p}) = \nabla_{p} V_{nav}(p)=0$ in view of \eqref{eqn:def_grad_bar_p_V}. Therefore, it follows from \eqref{eqn:critical_V_nav} that $p\in C_{V_{nav}}$ and $\theta = 0$ for all $(p,\theta)\in C_{\mathcal{V}_{nav}}$, which gives \eqref{eqn:def_V'_critical_set}. This completes the proof.

\subsection{Proof of Lemma \ref{lem:V_nav_synergistic}} \label{sec:V_nav_synergistic}
For the sake of simplicity, let us introduce $p' = \mathcal{T}(p,\bar{\theta})$, $d_{o1}=\|p-p_o\|=\|p'-p_o\|$ and $d_{o2}=\|p_o-p_d\|$ as shown in Fig. \ref{fig:diagram22}. From \eqref{eqn:def_V'_critical_set}, $(p,\theta)\in C_{\mathcal{V}_{nav}} \setminus \mathcal{A}_p$ implies that $\theta=0$ and $p\in C_{V_{nav}} \setminus \{p_d\}$ is located on the line $p_d-p_o$ and opposite to the destination point. \renewcommand{\thefigure}{\arabic{figure}}
\setcounter{figure}{4}
\begin{figure}[!ht]
	\centering
	\includegraphics[width=0.76\linewidth]{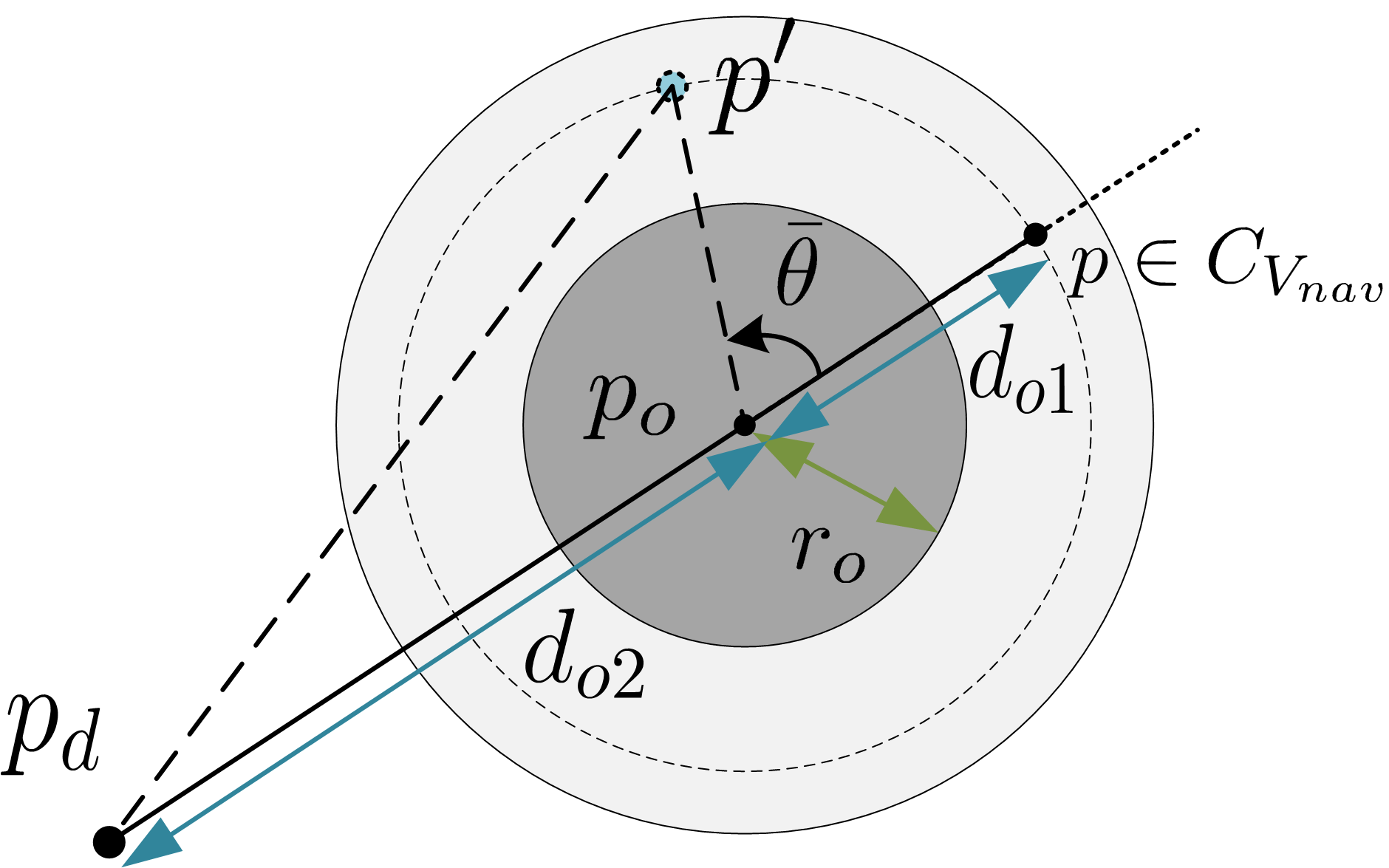}
	\caption{Geometric representation of the jump of the switching variable $\theta$ at the critical point $(p,\theta)\in C_{\mathcal{V}_{nav}}$.}
	\label{fig:diagram22}
\end{figure} 
For each $(p,\theta)\in C_{\mathcal{V}_{nav}}$ and $\bar{\theta}\in \Theta$, it follows from \eqref{eqn:def_V'_critical_set} that
\begin{align}
	& \mathcal{V}_{nav}(p,\theta) -  \mathcal{V}_{nav}(p,\bar{\theta}) \nonumber \\
	& = \frac{1}{2}\|\mathcal{T}(p,\theta) - p_d\|^2 + \varrho \phi(d_o(\mathcal{T}(p,\theta))) + \frac{\gamma_\theta}{2} \theta^2 \nonumber\\
	& \quad - \left(\frac{1}{2}\|\mathcal{T}(p,\bar{\theta}) - p_d\|^2 + \varrho \phi(d_o(\mathcal{T}(p,\bar{\theta}))) + \frac{\gamma_\theta}{2} \bar{\theta}^2 \right) \nonumber \\
	& = \frac{1}{2}\|p - p_d\|^2    - \frac{1}{2}   \left( \| \mathcal{T}(p,\bar{\theta}) - p_d\|^2  -  \gamma_\theta  \bar{\theta}^2 \right)   \label{eqn:V_nav_gap_step_1}
\end{align}
where we made use of the facts $C_{\mathcal{V}_{nav}} = C_{V_{nav}} \times \{0\}$, $\mathcal{T}(p,\theta=0) = p$ and $d_o(\mathcal{T}(p,\theta)) = d_o(\mathcal{T}(p,\bar{\theta})) = d_o(p)$ for all $\theta,\bar{\theta}\in \mathbb{R}$. Substituting the identities:
$\|p - p_d\|^2  =  d_{o1}^2  + d_{o2}^2 + 2 d_{o1}d_{o2} $ and 
$\|p' - p_d\|^2 =  d_{o1}^2 + d_{o2}^2   + 2 d_{o1} d_{o2} \cos(\bar{\theta})$ 
into  \eqref{eqn:V_nav_gap_step_1}, one obtains
\begin{align}
	\mathcal{V}_{nav}(p,\theta) -  \mathcal{V}_{nav}(p,\bar{\theta})  
	&  =  d_{o1} d_{o2} \left(1- \cos(\bar{\theta}) \right)- \frac{\gamma_\theta}{2} \bar{\theta}^2  \nonumber \\
	&  =  2d_{o1} d_{o2}\sin^2\left( \frac{\bar{\theta}}{2}\right)      - \frac{\gamma_\theta}{2} \bar{\theta}^2  \nonumber\\
	& > \frac{2r_o d_{o2} }{\pi^2} \bar{\theta}^2   - \frac{\gamma_\theta}{2} \bar{\theta}^2 \label{eqn:V_nav_gap_step_2}
\end{align}
where we made use of the facts $d_{o1} > r_o$, $1-\cos(\bar{\theta}) = 2\sin^2( \frac{\bar{\theta}}{2}) $, and $|\sin^2( \frac{\bar{\theta}}{2})| \geq  \frac{|\bar{\theta}|}{\pi}$ for all $|\bar{\theta}| \in (0,\pi)$. Similar to the proof of \cite[Proposition 2]{wang2022hybrid}, from the facts  $\gamma_\theta < \frac{4 r_o d_{o2} }{\pi^2}$ and $\bar{\theta}_M = \max_{\theta'\in \Theta} |\theta'|$,  one can further show that
\begin{align*}
	\mathcal{V}_{nav}(p,\theta) - \min_{\bar{\theta}\in \Theta} \mathcal{V}_{nav}(p,\bar{\theta})  
	&  =  \max_{\bar{\theta}\in \Theta} \left( \mathcal{V}_{nav}(p,\theta) - \mathcal{V}_{nav}(p,\bar{\theta})\right)  \nonumber \\
	&  >  \left( \frac{2r_o d_{o2} }{\pi^2}    - \frac{\gamma_\theta}{2}\right) \max_{\bar{\theta}\in \Theta}  \bar{\theta}^2 \nonumber\\
	&  =  \left( \frac{2 r_o d_{o2} }{\pi^2}    - \frac{\gamma_\theta}{2}\right)  \bar{\theta}_M^2 \nonumber \\
	&   = \delta_{\mathcal{V}}^* \geq \delta_{\mathcal{V}} %\label{eqn:V_nav_gap_step_3}
\end{align*}
with $d_{o2}=\|p_o-p_d\|$ and $\delta_{\mathcal{V}}^*= ( \frac{2 r_o \|p_d-p_o\|}{\pi^2}    - \frac{\gamma_\theta}{2} )  \bar{\theta}_M^2$. This completes the proof.

\subsection{Proof of Proposition \ref{prop:synergistic_single}}  \label{sec:synergistic_single}
To show that $(V,\kappa,\varpi,\Theta)$ is a synergistic feedback quadruple, we are going to verify all the conditions in Definition \ref{defn:synergistic_feedback} by considering  $V(x,\theta) = \mathcal{V}_{nav}(p,\theta)$ with $x = p\in \mathcal{X}_p = \mathcal{X}$. The conditions C1) and C2) in Definition \ref{defn:synergistic_feedback} can be verified for the quadruple $(V,\kappa,\varpi,\Theta)$ since $\mathcal{V}_{nav}$ in \eqref{eqn:def_synergistic_navigation} is a navigation function with respect to $\mathcal{A}$ by definition with the properties of the navigation function $V_{nav}$ in \eqref{eqn:def_navigation} and the transformation function  $\mathcal{T}$ in \eqref{eqn:def_mathcal_T}.  

Next, we are going to verify the conditions C3) and C4) in Definition \ref{defn:synergistic_feedback} for the quadruple $(V,\kappa,\varpi,\Theta)$. 
From \eqref{eqn:prop_kappa}, \eqref{eqn:single_integrator_modified}, \eqref{eqn:second_order_kappa} and \eqref{eqn:second_order_varpi}, for all $\mathcal{X}\times \mathbb{R}$ one has 
\begin{align*}
	& \langle \nabla_x V(x,\theta),  \kappa(x,\theta)\rangle
	+ \langle \nabla_\theta V(x,\theta), \varpi(x,\theta)\rangle\\
	&~ =\langle \nabla_p \mathcal{V}_{nav}(p,\theta),  -k_p \nabla_p \mathcal{V}_{nav}(p,\theta)\rangle
	- k_\theta |\nabla_\theta \mathcal{V}_{nav}(p,\theta)|^2 \\
	&~ = -k_p \|\nabla_p \mathcal{V}_{nav}(p,\theta)\|^2 - k_\theta |\nabla_\theta \mathcal{V}_{nav}(p,\theta)|^2 \leq 0
\end{align*}
which gives the condition C3) in Definition \ref{defn:synergistic_feedback} for the quadruple $(V,\kappa,\varpi,\Theta)$. Similar to \eqref{eqn:def_E}, we define the set $\mathcal{E} :=  \{(x,\theta)\in \mathcal{X} \times \mathbb{R}:
\nabla_p \mathcal{V}_{nav}(p,\theta)=0,  \nabla_\theta \mathcal{V}_{nav}(p,\theta)  = 0 \}
$,
which implies that the largest weakly invariant set for system \eqref{eqn:single_integrator_modified} with $(x,\theta) \in  \mathcal{E}$, denoted by $\Psi_{V} \subseteq \mathcal{E} $, is given as $\Psi_{V} = \mathcal{E} = C_{\mathcal{V}_{nav}}$ with $C_{\mathcal{V}_{nav}}$ defined in \eqref{eqn:def_V'_critical_set}.
From the definition of $\mu_{V,\Theta}$ in \eqref{eqn:mu_V} and the property of the synergistic navigation function $\mathcal{V}_{nav}$ as per Lemma \ref{lem:V_nav_synergistic}, one has
\begin{align}
	\mu_{V,\Theta}(x,\theta) & = V(x,\theta) - \min_{\bar{\theta}\in \Theta} V(x,\bar{\theta}) \nonumber\\
	& = \mathcal{V}_{nav}(p,\theta) - \min_{\bar{\theta}\in \Theta} \mathcal{V}_{nav}(p,\bar{\theta})  > \delta_{V} \label{eqn:mu_V_single}
\end{align}
for all $(x,\theta) \in  \Psi_{V}\setminus \mathcal{A}$. Then, from the definition of $\delta_{V,\Theta}$ in \eqref{eqn:mu_V_delta} and \eqref{eqn:mu_V_single}, one has
$
\delta_{V,\Theta} = \inf_{(x,\theta)\in \Psi_{V} \setminus \mathcal{A}}  \mu_{V,\Theta}(x,\theta)   > \delta_{\mathcal{V}}
$,
which verifies the condition C4) in Definition \ref{defn:synergistic_feedback} for the quadruple $(V,\kappa,\varpi,\Theta)$. Therefore, one concludes that the quadruple $(V,\kappa,\varpi,\Theta)$ is a synergistic feedback relative to  $\mathcal{A}$, with gap exceeding $\delta_\mathcal{V}$, for system \eqref{eqn:single_integrator_modified}.  This completes the proof.

\subsection{Proof of Lemma \ref{lem:sigma}}\label{sec:sigma}
From \eqref{eqn:def_V'_critical_set} and the proof of Proposition \ref{prop:synergistic_single}, one obtains  $\Psi_{V} = C_{\mathcal{V}_{nav}}= C_{V_{nav}}\times \{0\}$. 
It follows from \eqref{eqn:kappa_decomp2} that $\sigma(\theta) = 0$ for all $(x,\theta) \in \Psi_V$. Then, given a finite nonempty set $\Theta=\{|\theta|_i\in (0,\pi),i=1,\dots,L\} \subset \mathbb{R}$  one can show that for all  $(x,\theta) \in \Psi_V \setminus \mathcal{A}$
\begin{align*}
	\max_{\bar{\theta}\in \Theta } \|\sigma({\theta}) - \sigma(\bar{\theta})\|^2 
	& =  \max_{\bar{\theta}\in \Theta } \|\sigma(\bar{\theta})\|^2   \nonumber \\
	%&=   \|(I_2 - \mathcal{R}({\theta}))\T (p_d-p_o)\|^2  \nonumber\\
	%=   (p_d-p_o)\T (2I_2 - \mathcal{R}({\theta}) - \mathcal{R}({\theta})\T ) (p_d-p_o) 
	& = 2 \|p_d-p_o\|^2  \max_{\bar{\theta}\in \Theta } (1-\cos({\bar{\theta}}))  \nonumber \\
	& = 2(1-\cos(\bar{\theta}_M))  \|p_d-p_o\|^2   = 2c_k  %\label{eqn:c_kappa}
\end{align*}
where $\bar{\theta}_M = \max_{\bar{\theta}\in \Theta} |\bar{\theta}|$ and we made use of the facts  $\mathcal{R}({\theta})  \mathcal{R}({\theta})\T  = I_2$ and $2I_2 - \mathcal{R}({\theta}) - \mathcal{R}({\theta})\T   = 2(1-\cos({\theta}))I_2$. This completes the proof. 

\subsection{Proof pf Proposition \ref{prop:synergistic_single_smooth}} \label{sec:synergistic_single_smooth} 
The proof of Proposition  \ref{prop:synergistic_single_smooth} follows closely the proof of Proposition \ref{prop:SLFS} and \ref{prop:synergistic_single}, thus it will be abbreviated here. From the definition of $V_s$ in \eqref{eqn:def_V_s2}, one has
\begin{align} \label{eqn:nabla_V_s}
	\nabla V_s(x_s,\theta) & := \begin{pmatrix}
		\nabla_p V_s(x_s,\theta)\\ 
		\nabla_\eta V_s(x_s, \theta) \\
		\nabla_\theta V_s(x_s,\theta)
	\end{pmatrix} = \begin{pmatrix}
		\nabla_p \mathcal{V}_{nav}(p,\theta) \\ 
		\gamma_s(\eta - \sigma(\theta)) \\
		\nabla_\theta \mathcal{V}_{nav}(p,\theta)
	\end{pmatrix}.
\end{align}
From \eqref{eqn:second_order_kappa}, \eqref{eqn:second_order_varpi}, \eqref{eqn:single_integrator_modified2}, \eqref{eqn:def_kappa_s3} and \eqref{eqn:nabla_V_s}, one can show that
\begin{align}
	& \langle \nabla V_s(x_s,\theta), f_s(x_s,\theta) + g_s(x_s,\theta)\kappa_s(x_s,\theta)\rangle   \nonumber\\
	& = \langle \nabla V(x,\theta), f_c(x,\theta) + g_c(x,\theta)(\kappa_o(x)+k_p\eta)\rangle  \nonumber\\
	& \quad + \nabla_\eta V_s(x_s, \theta)\T (\kappa_s(x_s,\theta) - \mathcal{D}_t \sigma(\theta)) \nonumber\\
	& =   \nabla_p \mathcal{V}_{nav}(p,\theta)\T (-k_p \nabla_p V_{nav}(p,\theta)+ k_p \eta) \nonumber\\ 
	& \quad -   k_\theta |\nabla_\theta \mathcal{V}_{nav}(p,\theta)|^2 + \gamma_s\tilde{\eta}\T (\kappa_s(x_s,\theta) - \mathcal{D}_t \sigma(\theta)) \nonumber\\ 
	& =  - \nabla_p \mathcal{V}_{nav}(p,\theta)\T (-k_p \nabla_p \mathcal{V}_{nav}(p,\theta) + k_p \tilde{\eta})  \nonumber \\
	& \quad -   k_\theta |\nabla_\theta \mathcal{V}_{nav}(p,\theta)|^2 -\gamma_s  k_\eta  \|\tilde{\eta}\|^2   - k_p \tilde{\eta}\T    \nabla_p \mathcal{V}_{nav}(p,\theta) \nonumber\\
	& = - k_p \|\nabla_p \mathcal{V}_{nav}(p,\theta)\|^2  - k_\theta |\nabla_\theta \mathcal{V}_{nav}(p,\theta)|^2 - \gamma_s k_\eta \|\tilde{\eta}\|^2   \nonumber \\
	& \leq 0 \label{eqn:dobule_dot_Vs}
\end{align}
where we made use of the facts $\tilde{\eta} = \eta - \sigma(\theta)$ and $\kappa_o(x) + k_p \eta %= -k_p \nabla_p \mathcal{V}_{nav}(p,\theta) + k_p (\eta - \sigma(\theta)) 
= -k_p \nabla_p \mathcal{V}_{nav}(p,\theta) + k_p \tilde{\eta}$. 
Therefore, one has $\langle \nabla V_s(x_s,\theta), f_s(x_s,\theta) + g_s(x_s,\theta)\kappa_s(x_s,\theta)\rangle \leq 0 $, which implies the condition C3) in Definition \ref{defn:synergistic_feedback}  for the quadruple $(V_s,\kappa_s,\varpi,\Theta)$. Consequently, from \eqref{eqn:dobule_dot_Vs} the set $\mathcal{E}_s$ defined in \eqref{eqn:def_E_s} can be written as
$
\mathcal{E}_s =  \{(x_s,\theta)\in \mathcal{X}_s \times \mathbb{R}:
\nabla_p \mathcal{V}_{nav}(p,\theta)  = 0, \nabla_\theta \mathcal{V}_{nav}(p,\theta)  = 0, \eta = \sigma(\theta) \}
$. It follows from Lemma \ref{lem:grad_V_nav2} that $(p,\theta) \in C_{V_{nav}} \times \{0\}=C_{\mathcal{V}_{nav}}$ and $\eta = \sigma(0) = 0$ for all $(x_s,\theta)\in \mathcal{E}_s$. Hence, the largest weakly invariant set for the closed-loop system \eqref{eqn:new_affine_system} with $(x_s,\theta) \in  \mathcal{E}_s$ is given as
$
\Psi_{V_s}  = \{(x_s,\theta)\in \mathcal{X}_s \times \mathbb{R}: (p,\theta) \in C_{\mathcal{V}_{nav}}, \eta=0 \} 
$. 
From the definition of $V_s$ in \eqref{eqn:def_V_s2} and the results in \eqref{eqn:delta_V_s} and \eqref{eqn:mu_V_single},  one has
\begin{align*}
	\mu_{V_s,\Theta}(x_s,\theta)  & = V_s(x_s,\theta) - \min_{\bar{\theta}\in \Theta} V_s(x_s,\bar{\theta}) \nonumber\\
	& \geq \mu_{V,\Theta}(x,\theta)-\gamma_s c_k  >  \delta_{\mathcal{V}} - \gamma_s c_k > 0   %\label{eqn:mu_Vs_single}
\end{align*} 
for all $(x_s,\theta) \in  \Psi_{V_s}\setminus \mathcal{A}_s$ with $0<\gamma_s < \frac{\delta_{\mathcal{V}} }{c_k}$. It follows from the definition of $\delta_{V,\Theta}$ in \eqref{eqn:mu_V_delta} that
$
\delta_{V_s,\Theta}  = \inf_{(x_s,\theta)\in \Psi_{V_s} \setminus \mathcal{A}_s}  \mu_{V_s,\Theta}(x_s,\theta)   >  \delta_{\mathcal{V}}- \gamma_s c_k > \delta_{\mathcal{V}_s} %\label{eqn:mu_V_delta}
$, which implies the condition C4) in Definition \ref{defn:synergistic_feedback}  for the quadruple $(V_s,\kappa_s,\varpi,\Theta)$. Therefore, one concludes that  $(V_s,\kappa_s,\varpi,\Theta)$ is a synergistic feedback quadruple relative to  $\mathcal{A}_s$ for system \eqref{eqn:single_integrator_modified2} with gap exceeding $\delta_{\mathcal{V}_s}$.  This completes the proof.

%%%%%%%%%%%%%%%%%%%%%%%%%%%%%%%%%%%%%%%%%%%%%%%%%%%%%%%%%%%%%%%%%%%%%%%%%%%%%%%%

\bibliographystyle{IEEEtran}
\bibliography{mybib} 

\end{document}